\documentclass[format=acmsmall, review=false]{acmart}
\usepackage{acm-ec-24}
\usepackage[shortlabels]{enumitem}
\usepackage{booktabs} 
\usepackage{array}
\usepackage[ruled]{algorithm2e} 
\usepackage{longtable}

\SetAlFnt{\small}
\SetAlCapFnt{\small}
\SetAlCapNameFnt{\small}
\SetAlCapHSkip{0pt}
\IncMargin{-\parindent}

\newcommand{\paind}{p_{r_{\text{ind}}}}
\newcommand{\pfind}{p_{s_{\text{ind}}}}
\newcommand{\pastar}{p_{r^*}}
\newcommand{\pfstar}{p_{s^*}}

\newcommand{\prind}{p_{r_{\text{ind}}}}
\newcommand{\psind}{p_{s_{\text{ind}}}}
\newcommand{\prstar}{p_{r^*}}
\newcommand{\psstar}{p_{s^*}}
\newcommand{\pdagger}{p^\dagger}
\newcommand{\csstar}{c_{s^*}}

\newcommand{\pireq}{\pi_r^{\text{(eq)}}}
\newcommand{\piseq}{\pi_s^{\text{(eq)}}}

\newcommand{\pisell}{\pi_s^{(\ell)}}
\newcommand{\psstarell}{\psstar^{(\ell)}}

\newcommand{\alphasstar}{\alpha_{s^*}}
\newcommand{\alpharstar}{\alpha_{r^*}}
\newcommand{\alphardagger}{\alpha^\dagger_r}
\newcommand{\alphasdagger}{\alpha^\dagger_s}
\newcommand{\alphaopt}{\alpha_{s^*, r^*}}
\newcommand{\alphastar}{\alpha_*}
\newcommand{\alphastarell}{\alpha^{(o)}_*}
\newcommand{\alphamax}{\alpha_{\max}}

\newcommand{\Scal}{\mathcal{S}_p}

\newtheorem{lemma}{Lemma}
\newtheorem{proposition}{Proposition}

\newtheorem{corollary}{Corollary}

\newenvironment{proofsketch}{%
  \renewcommand{\proofname}{Proof sketch}\proof}{\endproof}

\setcitestyle{authoryear}

\title[A shared-revenue Bertrand game]{A shared-revenue Bertrand game}

\author{Raj Pabari}
\orcid{0009-0001-0968-9794}
\affiliation{%
\institution{Stanford}
\country{USA}
}
\email{rajpabari@stanford.edu}
\authornote{Correspondence to rajpabari@stanford.edu.}
\authornote{Work done in part while at Amazon.}

\author{Udaya Ghai}
\affiliation{%
\institution{Amazon}
\country{USA}
}

\author{Dominique Perrault-Joncas}
\affiliation{%
\institution{Amazon}
\country{USA}
}

\author{Kari Torkkola}
\affiliation{%
\institution{Amazon}
\country{USA}
}

\author{Orit Ronen}
\affiliation{%
\institution{Amazon}
\country{USA}
}

\author{Dhruv Madeka}
\authornotemark[2]
\affiliation{%
\institution{Google}
\country{USA}
}

\author{Aviad Rubinstein}
\affiliation{%
\institution{Stanford}
\country{USA}
}

\author{Dean Foster}
\affiliation{%
\institution{Amazon}
\country{USA}
}

\author{Omer Gottesman}
\orcid{0000-0003-4043-2556}
\affiliation{%
\institution{Amazon}
\country{USA}
}
\renewcommand{\shortauthors}{Pabari et al.}

\addtocontents{toc}{\protect\setcounter{tocdepth}{-1}}
\begin{abstract}
We introduce and analyze a variation of the Bertrand game in which the revenue is shared between two players. This game models situations in which one economic agent can provide goods/services to consumers either directly or through an independent seller/contractor in return for a share of the revenue. We analyze the equilibria of this game, and show how they can predict different business outcomes as a function of the players' costs and the transferred revenue shares. Importantly, we identify game parameters for which independent sellers can simultaneously increase the original player's payoff while increasing consumer surplus. We then extend the shared-revenue Bertrand game by considering the shared revenue proportion as an action and giving the independent seller an outside option to sell elsewhere. This work constitutes a first step towards a general theory for how partnership and sharing of resources between economic agents can lead to more efficient markets and improve the outcomes of both agents as well as consumers.
\end{abstract}
\begin{document}

\begin{titlepage}

    \maketitle

    \vspace{1cm}
    \addtocontents{toc}{\protect\setcounter{tocdepth}{-1}}
    \tableofcontents
    \addtocontents{toc}{\protect\setcounter{tocdepth}{1}}

\end{titlepage}

\section{Introduction}
\subsection{Motivation}
\label{sec:motivation}

Many modern businesses partner with independent service providers to increase efficiency. One method through which this is done is for a company to let a separate company provide services/goods to customers in exchange for a share of the revenue. The independent service provider in turn may benefit from the wider visibility it gains through access to the original company's pool of customers. This is seen in a number of real-world scenarios, including but not limited to:
\begin{enumerate}[(A)]
    \item Consider a large store which allows independent sellers to use some of its shelf space to sell products in return for a portion of the selling price. The large store may of course obtain the same products from a vendor and sell them on its own, and the independent seller may set up its own store, but depending on the costs for each party, both parties might benefit from having the independent seller's products available in the larger store.
    \item Consider a construction and maintenance company, with a large pool of customers, providing a variety of services such as plumbing, repairs, painting services, etc. The company can either provide a specific service on its own, or refer the customer to a contractor who sets the price for the service, and pays the original company a referral fee.
    \item Consider a ridesharing or taxi platform. When a passenger requests a ride, they have the option of transporting the passenger with either (i) a self-driving car, owned by the ridesharing platform, or (ii) an independent human driver, whom the ridesharing platform charges a fee for each passenger they refer.
    \item A large airline could either operate a short regional flight with their own aircraft or allow a regional partner to operate the flight.
    \item A healthcare provider can either provide specialist services in-house or refer their patients to an external contractor, charging a referral fee.
\end{enumerate}
\subsection{Overview}

In this work we provide a game theoretic model for reasoning about the interaction between agents in scenarios such as these with a variation of the Bertrand pricing game \cite{bertrand1883review}. For concreteness, we will focus our intuition and analysis on scenario (A) from Section \ref{sec:motivation}, though our model can explain phenomena in domains far beyond retail sales. In accordance with scenario (A), the ``shared-revenue Bertrand game'' is played between a \emph{retailer (r)}, and an independent goods/services provider which we will refer to as the \emph{independent seller (seller, s)}. Both players sell an identical good/service, with each agent having a different cost. Under the revenue sharing program, whenever the independent seller fulfills demand, a portion of the revenue, called the \emph{referral fee $(\alpha)$}, is transferred from the seller to the retailer. Specifically, the game is as follows:

\label{game:shared-revenue-bertrand}
\begin{center} {\bf Shared-Revenue Bertrand Game}\smallskip
    \hrule
\end{center}{\bf Parameters}\\
Demand curve $q(p)$, retailer's cost $c_r$, seller's cost $c_s$\footnote{The leaving subgame of Section \ref{sec:leaving-subgame} is additionally parameterized by the outside option cost differential $\delta\in (0, c_r-c_s)$.}

\noindent {\bf Gameplay}\begin{enumerate}
    \item The retailer chooses a referral fee $\alpha \in (0, 1)$ (Section \ref{sec:alpha-game-setup})
    \item The independent seller chooses one of two options (Section \ref{sec:outside-option}):
          \subitem Stay in the retailer's revenue sharing program
          \subitem Leave the retailer's revenue sharing program and sell on their own
    \item The retailer and seller choose their selling prices:
          \subitem If the independent seller chose to stay, they play a staying subgame (Sections \ref{sec:preliminaries}-\ref{sec:refining-equilibria})
          \subitem If the independent seller chose to leave, they play a leaving subgame (Section \ref{sec:leaving-subgame})
\end{enumerate}\hrule\bigskip
\noindent The primary contribution of this work is in the half of the game tree where the independent seller stays in the revenue sharing program and the firms compete with coupled payoff functions. In Section \ref{sec:preliminaries}, we set up the staying game in more detail, then in Section \ref{sec:equilibria_analysis} we solve for the Nash equilibria of this subgame. In Section \ref{sec:refining-equilibria}, we interpret and analyze the equilibria of the staying subgame in depth, examining the effects of the revenue sharing program on the firms' payoffs and the end price to the customer.

The referral fee is a constant parameter of the staying subgame of Section \ref{sec:preliminaries}. In Section \ref{sec:alpha-game-setup} we proceed by analyzing the ``fee optimization game,'' which is the same as the the shared-revenue Bertrand game except the independent seller does not have the option to leave the revenue sharing program. The equilibria of the fee optimization game reveal that the retailer always prefers to allow the independent seller to fulfill demand in equilibrium, setting the referral fee strategically so as to maximize their payoff. There is a sweet spot for the optimal fee such that it is low enough to stimulate economic activity while high enough to maximize payoff. However, we still observe some unrealistic equilibria in which the retailer sets such a high fee that the independent seller must fulfill demand at their effective cost.

Finally, in Section \ref{sec:outside-option}, we solve the entire shared-revenue Bertrand game, reintroducing the ability for the independent seller to leave the revenue sharing program after observing the retailer's choice of referral fee. In Section \ref{sec:leaving-subgame}, we choose the leaving subgame to be one of standard Bertrand competition with asymmetric costs. Since the literature on direct price competition is extensive, we do not focus closely on choosing the leaving subgame \cite{bertrand1883review,shubik1980market,tirole1988theory}. With this, if the independent seller's threat of selling on their own is credible, the unrealistic equilibria of the fee optimization game will no longer be observed. The retailer must instead set the referral fee low enough to ensure the independent seller doesn't exercise their outside option.

Importantly, our analysis demonstrates that the retailer will often be incentivized to set a fee at which both the independent seller increases its payoff compared to selling independently, and the price to the customer is decreased relative to the price either the retailer or the independent seller would have set independently. Our results therefore suggest that such a revenue sharing mechanism can not only increase social welfare, but also simultaneously improve outcomes for both players as well as consumers.

\subsection{Related work}
Our model is an extension of the standard Bertrand game \cite{bertrand1883review} in which two sellers of an identical product with the same cost set their prices independently. As we justify in Appendix \ref{appendix:justification_of_setup}, we consider the regime in which the sellers' costs are asymmetric, which is a natural extension of the standard Bertrand game, studied for instance in \cite{blume2003, kartik2011, demuynck2019}. Prior works have pursued other modifications to the Bertrand game to better describe platform competition, for instance in \cite{rochettirole2003,angelini2024,armstrong2006}.

One field in which revenue sharing programs as described in Section \ref{sec:motivation} have been studied is in the independent contractor literature. Much of this literature analyzes relationships with independent contractors through the framework of principal agent theory \cite{coats2002, chang2014}. The revenue-sharing program can also be seen as a form of ``make-or-buy'' decision for the retailer \cite{coase1937, williamson1975markets, tadelis2002}. Some revenue sharing programs are implemented via two-part tariffs which involve a fixed base fee alongside the per-unit fee \cite{yin2004, griva2015}. A final field in which revenue sharing programs have been studied is optimal taxation theory, where the referral fee can be viewed as the retailer imposing a tax on the independent seller. Viewed in this light, our results about optimal referral fees resonate strongly with optimal taxation theories such as \cite{laffer2004,stern1976,stiglitz1987}.

Finally, the idea of an outside option that we introduce in Section \ref{sec:outside-option} is common in negotiation and bargaining theory \cite{binmore1989,osborne1990bargaining}, which makes sense because the participation of the independent seller in the retailer's revenue-sharing program can be seen as a negotiation of a contract. In fact, there is even an ``outside option principle'' \cite{watson2019}, which informally states than an outside option must be credible in order to have an effect on the Nash equilibria of the game; our results are consistent with this.

\section{Preliminaries}
\label{sec:preliminaries}
The core staying subgame is one in which two players -- the retailer ($r$), and the independent seller ($s$), sell an identical product. The actions in the game are the price each player chooses to charge, and we assume the players move simultaneously. Formally, we define the following game:

\label{game:staying-subgame}
\begin{center} {\bf Staying Subgame}\smallskip
    \hrule
\end{center}{\bf Parameters}\\
Demand curve $q(p)$, referral fee $\alpha$, retailer's cost $c_r$, seller's cost $c_s$

\noindent {\bf Gameplay}\begin{enumerate}
    \item The retailer and seller simultaneously choose prices, yielding a price configuration $(p_r, p_s)$
\end{enumerate}
\hrule\bigskip
The payoffs for the players are
\begin{align}
    \label{eq:a_payoff}
    \overline{\pi}_r & = (\overline{p}_r - \overline{c}_r) \overline{q}_r + \alpha \overline{p}_s \overline{q}_s \\
    \label{eq:f_payoff}
    \overline{\pi}_s & = (\overline{p}_s - \overline{c}_s - \alpha \overline{p}_s) \overline{q}_s
\end{align}
%
%
where $\overline{p}$, $\overline{c}$, $\overline{q}$ and $\alpha$ denote prices, costs, quantities sold and the referral fee, respectively, and subscripts denote the players. This game can be viewed as a modification of the Bertrand game, where the players have different costs, and a portion of the independent seller's revenue is shared with the retailer as part of the service provided by the retailer.

We will suppose that the demand curve $\overline{q}(\overline{p})$ is strictly concave, compactly supported, monotonically decreasing, twice-differentiable, and nonnegative\footnote{The only ``real" assumptions here are concavity and continuity. Compact support is imposed for convenience so we can normalize to $[0,1]$. If monotonicity and/or strictness are violated, the key prices and fees may become nonunique (connected intervals). If $q$ is only once-differentiable, we can replace derivatives with subderivatives.}. When both players set the same price, we also assume they split the market where the retailer fulfills $\beta \in [0, 1]$ of the demand and the independent seller fulfills the rest. This assumption is made for mathematical completeness, although in Proposition~\ref{thm:no_split_market} we show that no split market equilibrium exists, making the choice of $\beta$ irrelevant. Concretely, the retailer's demand curve is given by
\begin{align}
    \label{eq:demand_curve}
    \overline{q}_r(\overline{p}_r, \overline{p}_s) = \begin{cases}
                                                         \overline{q}(\overline{p}_r)       & \overline{p}_r < \overline{p}_s \\
                                                         \beta \overline{q}(\overline{p}_r) & \overline{p}_r = \overline{p}_s \\
                                                         0                                  & \overline{p}_r > \overline{p}_s
                                                     \end{cases}
\end{align}
The demand curve for the independent seller is analogous, but $\beta$ is replaced by $(1 - \beta)$. Because $q$ is compactly supported, there exists a maximum price $\overline{p}_{\max}\equiv \inf\{\overline{p}\mid\overline{q}(\overline{p})=0\}$. Throughout the rest of the paper we present the results in terms of unitless variables
\begin{align}
    \label{eq:normalization}
    p =  \frac{\overline{p}}{\overline{p}_{\max}}\ ; \quad c = \frac{\overline{c}}{\overline{p}_{\max}}\ ; \quad q(p)=\frac{\overline{q}(p)}{\overline{q}(0)}\ ; \quad \pi(p_r, p_s) = \frac{\overline{\pi}(p_r,p_s)}{\overline{q}(0)\overline{p}_{\max}}\ .
\end{align}
Under this normalization, the demand at price $p=0$ is 1, and the demand at price $p=1$ is $0$. Equations \eqref{eq:a_payoff}, \eqref{eq:f_payoff}, and \eqref{eq:demand_curve} remain unchanged if we remove bars. Note that under this normalization, $q$ inherits the properties of $\overline{q}$ (concavity, monotonicity, differentiability, and nonnegativity).

\subsection{Necessary assumptions}
\label{sec:assumptions}
In order for the game to be nontrivial, we must impose the following constraints:
\begin{enumerate}[(I)]
    \item There exists a price at which both the retailer and independent seller could achieve a nonzero payoff when fulfilling demand ($0 < c_s, c_r < 1$).
    \item The seller's cost is less than the retailer's cost ($c_s < c_r$).
    \item The referral fee takes some, but not all of the independent seller's revenue ($\alpha\in (0,1)$).
\end{enumerate}
In Appendix \ref{appendix:justification_of_setup}, we formally show that the game degenerates if these are violated, but the importance of each condition should be fairly intuitive. Overall, the revenue-sharing program should be thought of as for products that the retailer could conceivably sell themselves, but the independent seller is able to do so more efficiently.

%
\subsection{Important payoffs and prices}
\label{sec:additional_notation}


We now introduce several important prices which will be used throughout this paper. We also include a full notation table in Appendix \ref{appendix:notation_table}, and include derivations of the key prices and fees in Appendix \ref{appendix:key-prices}.

As we will later prove in Proposition~\ref{thm:no_split_market}, we only need to focus on cases where either the retailer or the independent seller fulfill all demand, and can ignore any split market scenario. When we write a payoff $\pi$ with two subscripts, the first denotes the the player whose payoff we refer to, and the second denotes the player who fulfills all demand. Thus, $\pi_{r, r}(p_r)$ and $\pi_{s, r}(p_r)$ are the retailer's and the independent seller's payoffs when the retailer fulfills all the demand, respectively. Similarly $\pi_{r, s}(p_s)$ and $\pi_{s, s}(p_s)$ are the retailer's and the independent seller's payoffs when the independent seller fulfills all the demand, respectively. Explicitly, it follows from \eqref{eq:a_payoff}, \eqref{eq:f_payoff}, and \eqref{eq:demand_curve} that
\begin{align}
     & \pi_{r,r}(p_r) =(p_r-c_r)q(p_r)\ ;   &  & \pi_{s, r}(p_r) = 0 \label{eq:retailer-fulfills-demand-payoffs} \\
     & \pi_{r,s}(p_s) = \alpha p_sq(p_s)\ ; &  & \pi_{s, s}(p_s) = ((1-\alpha)p_s - c_s)q(p_s)
    \label{eq:seller-fulfills-demand-payoffs}
\end{align}
In Figure~\ref{fig:payoff_example} we plot the payoffs of the players as a function of price for two choices of game parameters $(c_r, c_s, \alpha)$ in the special case that $q(p)=1-p$ is linear. All subsequent figures throughout the paper will also use this linear demand curve for consistency. Solid and dashed lines denote the players' payoffs when the retailer or independent seller fulfill all demand, respectively, while blue and red lines denote the payoffs to the retailer and the independent seller, respectively.

\begin{figure}
    \begin{center}
        \includegraphics[width=0.99
            \textwidth]{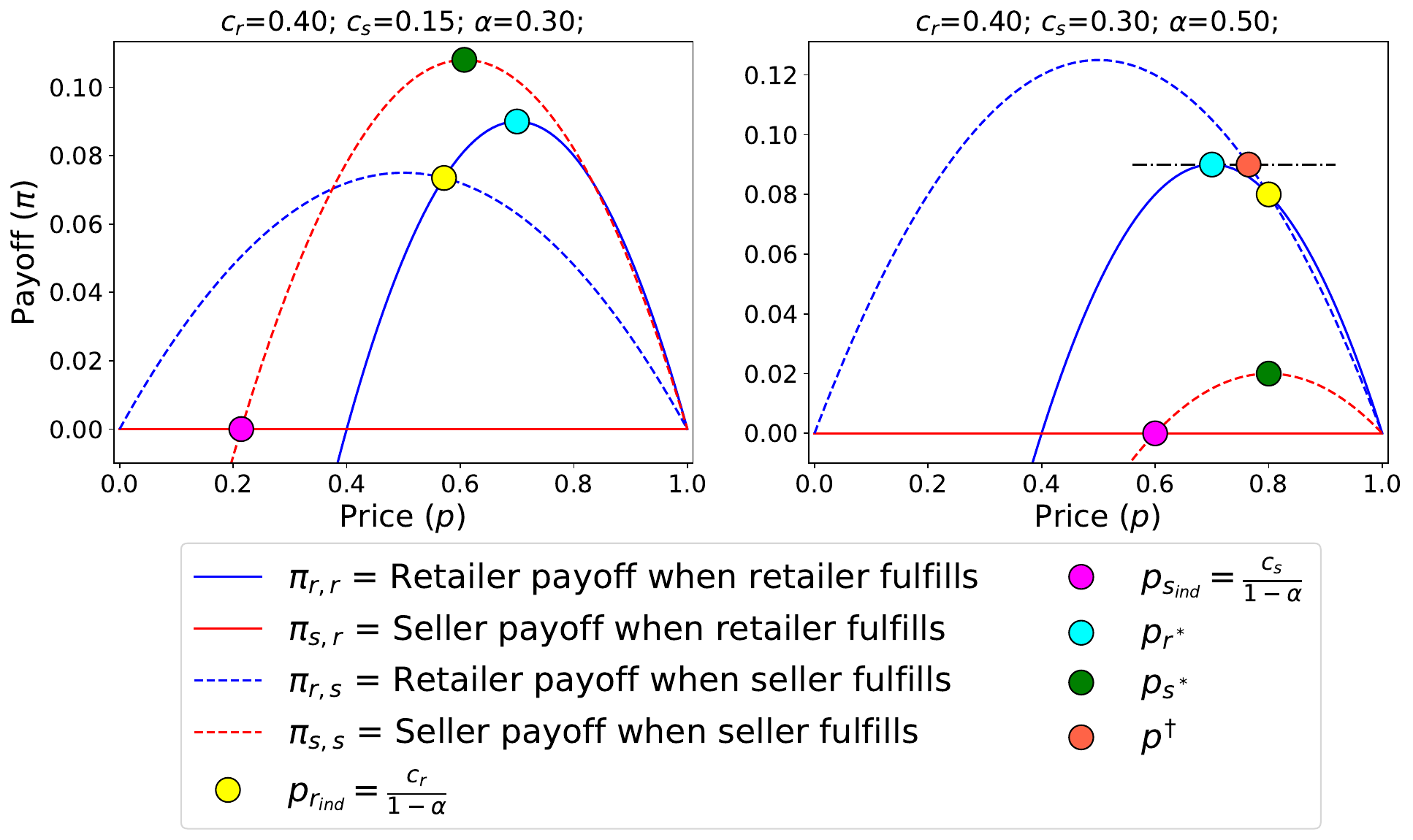}
        \caption{Sample payoff curves and important price points for linear demand curve $q(p) = 1-p$}
        \label{fig:payoff_example}
    \end{center}
\end{figure}

Figure~\ref{fig:payoff_example} also illustrates five important price points and their corresponding payoffs which will be used throughout this paper. We denote by $\pastar$ and $\pfstar$ the prices each player would set to maximize their own payoff if they fulfilled all demand, and call these the retailer's or independent seller's optimal selling price. Differentiating the players' single agent payoffs $\pi_{r,r}$ and $\pi_{s,s}$ and equating to zero gives first-order conditions for these prices in \eqref{eq:prstar-foc} and \eqref{eq:psstar-foc}.

We denote by $\paind$ and $\pfind$ the prices at which the retailer and the independent seller are indifferent to which player fulfills demand, i.e. $\pi_{r, r}(\paind) = \pi_{r, s}(\paind)$ and $\pi_{s, r}(\pfind) = \pi_{s, s}(\pfind)$. Solving these equations in \eqref{eq:prind-derivation} and \eqref{eq:psind-derivation} yields the following prices and payoffs --
\begin{align}
     & \paind \equiv \frac{c_r}{1 - \alpha}\ ; &  & \psind\equiv \frac{c_s}{1-\alpha}
\end{align}
We first note that $\pi_{s, s}(\pfind) = 0$, and therefore the independent seller's indifference price is also its breakeven price. More interestingly, while the retailer's breakeven price is still its cost, $c_r$, its indifference price is scaled by $(1 - \alpha)^{-1}$, but at that price its payoff is nonzero. This is the main feature which differentiates the shared-revenue Bertrand game from the classical Bertrand game, as the retailer can choose to fulfill demand on its own or allow the independent seller to fulfill demand, making revenue from the fee. As a result, this game gives rise to a more complex equilibrium structure which depends on the fees and costs of the players. Note also that because we assume $c_s < c_r$ (Section \ref{sec:assumptions}), we have $\pfind < \paind$.

Next, we introduce the price at which the retailer's payoff when the independent seller fulfills demand is equal to its optimal payoff had it fulfilled demand on its own, i.e. $\pdagger$ is the largest price such that
\begin{align}
    \pi_{r, r}(\pastar) = \pi_{r, s}(p^\dagger)
\end{align}
Note that $p^\dagger$ is only well-defined when $\alpha$ is sufficiently large; for instance it is not well-defined for the parameters in the left subplot of Figure \ref{fig:payoff_example} by inspection.
\section{Equilibrium analysis}
\label{sec:equilibria_analysis}
\subsection{Market split}
For simplicity, we will only consider pure strategies in this paper; thus, a strategy profile can be characterized fully by a tuple $(p_r, p_s)\in \Scal \equiv (0,1)\times (0,1)$ denoting the prices chosen by the retailer and independent seller. We first prove, in Proposition~\ref{thm:no_split_market}, that in all equilibria, one player takes the entire market.
\begin{proposition}
    \label{thm:no_split_market}
    Let $(p_r,p_s) = (p,p)$ be an equilibrium price configuration, with the equilibrium quantities sold by each player $q_r = \beta q(p)$ and $q_s = (1 - \beta) q(p)$. Then $\beta \in \{ 0, 1 \}$.
\end{proposition}
\begin{proofsketch}
    To prove the Proposition we show that when the market is split, the retailer's payoff is a weighted average of $\pi_{r, r}$ and $\pi_{r, s}$, and one of them is larger than the split market payoff. Thus, the retailer can always increase their payoff by either lowering or raising the price by an infinitesimal amount, fulfilling demand themselves or allowing the independent seller to do so respectively. See Appendix~\ref{appendix:no_split_market_proof} for the complete proof.
\end{proofsketch}
While this result is not surprising given that in nearly all standard non-symmetric game theory models, in equilibrium only one player fulfills all demand, this result is clearly not expected to hold in the real world. Game theory results may mimic the world more closely and arrive at more realistic equilibria in which players split demand by modeling factors such as incomplete information and brand loyalty, but these are beyond the scope of this work.

\subsection{Nash equilibrium derivation}
In light of Proposition~\ref{thm:no_split_market}, we can denote all equilibria in the game by either unequal prices $(p_r, p_s) \in \Scal$ or a shared price $(p, x)\in (0,1)\times \{r,s\}$, where $x$ denotes the player who fulfills the demand. In this subsection, we aim to provide an intuitive sketch of the Nash equilibria.

To narrow our focus, let's consider shared-price equilibria $(p,s)$ where the independent seller fulfills demand at price $p$ for now. In order for $(p, s)$ to be an Nash equilibrium, the following conditions must be satisfied:
\begin{enumerate}
    \item 0 = $\pi_{s, r}(p) \leq \pi_{s, s}(p)$
    \item $\sup_{p'\in [0, p) }\left[\pi_{s,s}(p')\right]\leq \pi_{s,s}(p)$
    \item $\sup_{p' \in [0, p)} \left[ \pi_{r, r} (p') \right] \leq \pi_{r, s}(p)$
\end{enumerate}
If condition 1 doesn't hold, the independent seller can achieve a higher payoff by increasing its price and letting the retailer fulfill demand. If condition 2 doesn't hold, the independent seller can increase its payoff by lowering their price and continuing to fulfill all demand. If condition 3 doesn't hold, then there is some price lower than $p$ the retailer can set to take the entire market and increase its payoff.

\begin{figure}
    \begin{center}
        \includegraphics[height=5.5cm]{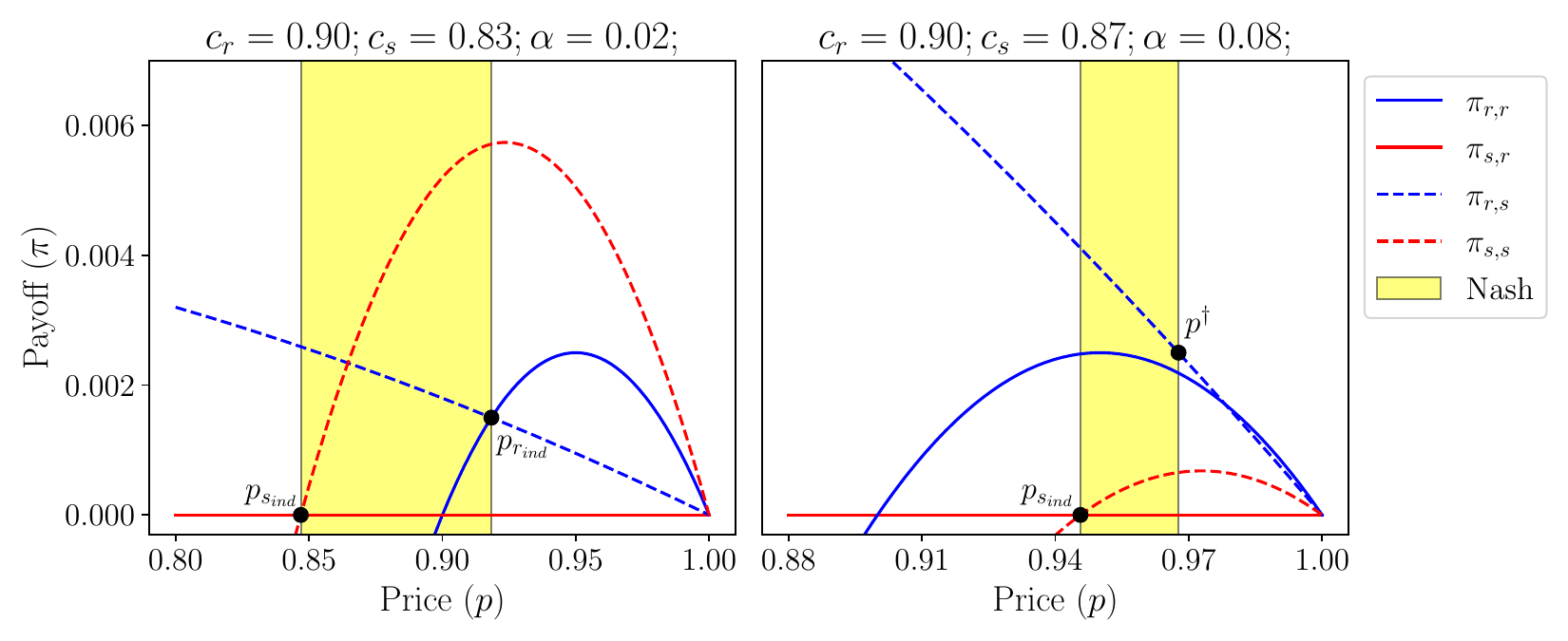}
        \caption{Shared-price Nash equilibria where the independent seller fulfills demand}
        \label{fig:nash-equilibrium-payoffs}
    \end{center}
\end{figure}
In Figure \ref{fig:nash-equilibrium-payoffs}, we highlight in yellow the Nash equilibrium prices where the independent seller fulfills demand. In the left subplot of Figure \ref{fig:nash-equilibrium-payoffs}, the Nash equilibria occur at prices $p\in [\psind, \prind]$, which satisfies all of the Nash equilibrium conditions:
\begin{enumerate}
    \item The seller achieves nonnegative payoff because $p\geq \psind$.
    \item The independent seller's payoff would decrease if they set a price $p'<p$ because $p\leq \psstar$.
    \item The retailer would not prefer to fulfill demand themselves because $p\leq \prind$.
\end{enumerate}
In the right subplot of Figure \ref{fig:nash-equilibrium-payoffs}, the Nash equilibria occur at prices $p\in [\psind, p^\dagger]$. Just as in the left subplot, conditions 1 and 2 are satisfied because $p\geq \psind$ and $p\leq \psstar$. Condition 3 is satisfied because $p\leq p^\dagger$. Notice that for the payoff curves in the right subplot, this is a stronger condition than $p\leq \prind$. This stronger condition is necessary because if $p> p^\dagger$, the retailer would prefer to set price $\prstar < p$ and achieve a payoff of
\[\pi_{r,r}(\prstar)=  \pi_{r,s}(p^\dagger)> \pi_{r,s}(p)\]
In Appendix \ref{appendix:shared-price-is}, we carry forth this intuition and translate conditions 1-3 into feasible intervals of prices for any $\alpha\in (0,1)$ rather than just for the specific referral fees in Figure \ref{fig:nash-equilibrium-payoffs}. For $(p,s)$ to be a Nash equilibrium, we find that the price must satisfy
\begin{equation}p\in \left[\psind, \min\{\prind,\prstar, \psstar\}\right]\cup \left[\max\{\prstar, \psind\}, \min \{p^\dagger, \psstar\}\right]\label{eq:price-conditions-nash-equilibrium-seller}\end{equation}
We acknowledge that it's still not clear what the equilibrium interval of prices look like for each possible choice of $\alpha$, namely what lower and upper bounds are attained for each possible referral fee -- this warrants a longer discussion which we defer to Section \ref{sec:nash-equilibrium-existence}. For instance, $p^\dagger$ is not well-defined for small $\alpha$, but we will see in \eqref{eq:alphardagger-introduction-again} and Lemma \ref{lemma:prind-prstar-pdagger-ordering} that whenever $\left[\max\{\prstar, \psind\}, \min \{p^\dagger, \psstar\}\right]$ is nonempty, $\alpha$ is high enough for $p^\dagger$ to be well-defined. For now, simply keep in mind that \eqref{eq:price-conditions-nash-equilibrium-seller} is sufficient only for shared-price Nash equilibria where the independent seller fulfills demand.

In Appendix \ref{appendix:shared-price-retailer}, we derive the conditions analogous to \eqref{eq:price-conditions-nash-equilibrium-seller} for shared-price equilibria $(p,r)$ where the retailer fulfills demand and find the following:
\begin{equation}
    p\in \left[\prind, \min\{ \psind, \prstar, \psstar\}\right]\label{eq:price-conditions-nash-equilibrium-retailer}
\end{equation}
Upon further inspection of \eqref{eq:price-conditions-nash-equilibrium-retailer}, we notice a subtle problem. Recall that $\psind < \prind$ always because $c_s < c_r$. Thus, the interval $\left[\prind, \min\{ \psind, \prstar, \psstar\}\right]$ is empty for all $\alpha$, which implies \emph{there are no shared price-equilibria where the retailer fulfills demand!}

This completes the first step of our analysis for shared-price equilibria, but we also need to consider unequal-price equilibria. Indeed, we consider equilibria where $p_r <p_s$ in Appendix \ref{appendix:unequal-price-retailer} and equilibria where $p_s <p_r$ in Appendix \ref{appendix:unequal-price-is}. Interested readers can find conditions similar to \eqref{eq:price-conditions-nash-equilibrium-seller} and \eqref{eq:price-conditions-nash-equilibrium-retailer} in the corresponding appendices; we also include a summary of all Nash equilibria with both equal and unequal prices in Section \ref{sec:summary-of-nash-equilibria}.

\subsection{Important fees for transitions between Nash equilibria}
\label{sec:nash-equilibrium-existence}
In this section, we answer the aforementioned question of what the interval in \eqref{eq:price-conditions-nash-equilibrium-seller} looks like for different values of $\alpha$. To do this, we will introduce five threshold fees that delineate the transitions between when different Nash equilibria can occur (skip ahead to Figure \ref{fig:param_space} for a preview). Similar analysis as in this section for unequal price equilibria can be found in Appendices \ref{appendix:unequal-price-retailer} and \ref{appendix:unequal-price-is}.

Consider first the left interval of \eqref{eq:price-conditions-nash-equilibrium-seller}, $p\in \left[\psind, \min\{\prind,\prstar, \psstar\}\right]$. We know that $\psind <\prind$ always. We have $\psind \leq \psstar$ so long as $\psind \leq 1$, which occurs when $\alpha\leq 1-c_s$. The condition for $\psind\leq \prstar$ is not quite as clear, so we introduce the retailer optimal price feasibility fee $\alpharstar$ (derived in \eqref{eq:alpharstar-derivation}) --
\begin{align}
     & \psind\leq \prstar \iff \alpha \leq \alpharstar\ ; &  & \text{where }\alpharstar \equiv 1-\frac{c_s}{\prstar}\label{eq:alpharstar-introduction}
\end{align}
The name comes from the fact that $\prstar\leq \psind$ is necessary for any equilibrium where the retailer fulfills demand at their optimal price, which we prove in Appendix \ref{appendix:unequal-price-retailer}. Note also that $\alpharstar\leq 1-c_s$ since $\prstar \in (0,1)$. Thus, the left interval of \eqref{eq:price-conditions-nash-equilibrium-seller} is nonempty if and only if $\alpha \leq \alpharstar$.

Now, we turn our attention to the value of $\min\{\prind,\prstar, \psstar\}$, the upper bound of the left interval of \eqref{eq:price-conditions-nash-equilibrium-seller}. This motivates the optimality switching fee $\alphaopt$ (derived in \eqref{eq:alphaopt-derivation}), the seller optimal price feasibility fee $\alphasstar$ (derived in \eqref{eq:alphasstar-derivation}), and the retailer $p^\dagger$ relevance fee $\alphardagger$ (derived in \eqref{eq:alphardagger-derivation}) --
\begin{align}
     & \psstar\leq \prstar \iff \alpha \leq \alphaopt\ ;     &  & \text{where }\alphaopt \equiv 1-\frac{c_s}{c_r} \label{eq:alphaopt-introduction}            \\
     & \prind\leq \psstar \iff \alpha \leq \alphasstar\ ;    &  & \text{where }\alphasstar \equiv 1-\frac{c_r}{\psstar}\label{eq:alphasstar-introduction}     \\
     & \prind \leq \prstar \iff \alpha \leq \alphardagger\ ; &  & \text{where }\alphardagger \equiv 1-\frac{c_r}{\prstar}\label{eq:alphardagger-introduction}
\end{align}
Analogously to $\alpharstar$, we will soon find that $\alpha\geq \alphasstar$ is necessary for an equilibrium where the independent seller fulfills demand at $\psstar$. Intersecting the inequalities in \eqref{eq:alpharstar-introduction}, \eqref{eq:alphaopt-introduction}, \eqref{eq:alphasstar-introduction}, and \eqref{eq:alphardagger-introduction} is sufficient to determine what the left interval of \eqref{eq:price-conditions-nash-equilibrium-seller} looks like in each region of parameter space. We will omit the algebra here and instead opt to summarize this in Section \ref{sec:summary-of-nash-equilibria}. However, throughout the course of the derivation in Appendix \ref{appendix:equilibria_derivation}, an important condition emerges that determines the relative ordering of $\alphaopt$, $\alphasstar$, and $\alphardagger$. We show in \eqref{eq:selleroptcost-derivation-i} and \eqref{eq:selleroptcost-derivation-ii} that this seller optimal feasibility cost $\csstar$ satisfies the following --
\begin{align}
     & \alphasstar\leq \alphardagger \iff \alphardagger \leq \alphaopt\iff c_s\leq \csstar\ ; &  & \text{where }\csstar \equiv \frac{c_r^2}{\prstar}
\end{align}
Finally, let's turn our attention to the right interval of \eqref{eq:price-conditions-nash-equilibrium-seller}, $p\in  \left[\max\{\prstar, \psind\}, \min \{p^\dagger, \psstar\}\right]$. Fortunately, we have already determined the relative ordering of each of the prices in the right interval, except for $p^\dagger$. For these conditions involving $p^\dagger$, we introduce the $p^\dagger$ relevance fees for the retailer and seller (derived in \eqref{eq:alphardagger-derivation} and Corollary \ref{cor:alphasdagger-ordering}):
\begin{align}
     & p^\dagger\leq \prstar \iff \alpha \leq \alphardagger\ ;                                                  &  & \text{where }\alphardagger \equiv 1-\frac{c_r}{\prstar}\label{eq:alphardagger-introduction-again} \\
     & \psind\leq p^\dagger \iff \alpha \leq \alphasdagger\ ;\footnotemark \label{eq:alphasdagger-introduction} &  &
\end{align}
\footnotetext{This equivalence only holds when $\alpha$ is sufficiently high (see Corollary \ref{cor:alphasdagger-ordering} for details). Because of this, instead of defining $\alphasdagger$ via the equivalence \eqref{eq:alphasdagger-introduction}, we instead construct $\alphasdagger$ as the largest fee such that $\pi_{r,s}(\psind, \alphasdagger) = \pi_{r,r}(\prstar)$.}
Note that we already introduced $\alphardagger$ in \eqref{eq:alphardagger-introduction}. The notation arises because we will show that $\alphardagger \leq \alpha \leq \alphasdagger$ is necessary for any equilibrium where demand is fulfilled at $p^\dagger$. We cannot find a closed form for $\alphasdagger$ in general, though in Appendix \ref{appendix:relative-ordering-fees} we are still able to produce the bound $\alphasdagger\in [\max\{\alpharstar,\alphardagger, \alphaopt\}, 1-c_s]$. Finally, the relative ordering of $\psstar$ and $p^\dagger$ is complex, and there unfortunately is no special $\hat{\alpha}$ for which $\psstar\leq p^\dagger\iff \alpha\leq\hat{\alpha}$. Thus, we will leave that condition unsimplified for now.

\subsection{Summary of Nash equilibria}
\label{sec:summary-of-nash-equilibria}
Considering all possible values of $\alpha$ and all possible price configurations, we summarize the Nash equilibria of the game below. We denote sets of equilibrium price configurations as $(S_r, S_s)\subset \Scal$ and singletons without the set notation, specifically $x\equiv \{x\}$. Recall there are no shared-price equilibria where the retailer fulfills demand, thus all sets of shared-price equilibria $(S_r, S_s) = (p_s, [\cdot, \cdot])$ implicitly assume that the independent seller is fulfilling demand at price $p_s\in [\cdot, \cdot]$, not the retailer.

\begin{center} {\bf Equilibrium Price Configurations of \hyperref[game:staying-subgame]{Staying Subgame}}\smallskip
    \hrule
\end{center}{\bf Notation}\\
We denote Nash equilibrium strategies by price configurations $(p_r, p_s)\in (S_r, S_s)\subset\Scal$.

\noindent\textbf{Case (i)}, $c_s \leq \csstar$:
\begin{equation}
    (p_r, p_s)\in \begin{cases}
        (p_s, \left[\psind, \prind\right])                                        & \alpha
        \leq \alphasstar                                                                                                                                                         \\
        (\left[\psstar,1\right],\psstar)\cup (p_s, \left[\psind,\psstar\right])   & \alphasstar\leq \alpha \leq \alphaopt                                                        \\
        (\left[\psstar, 1\right], \psstar)\cup (p_s, \left[\psind,\psstar\right]) & \left[\alphaopt\leq \alpha\leq \alphasdagger\right]\land \left[\psstar \leq p^\dagger\right] \\
        (p_s,\left[\psind,p^\dagger\right])                                       & \left[\alphaopt\leq \alpha \leq \alphasdagger\right]\land \left[p^\dagger\leq \psstar\right] \\
        (\prstar, \left[p^\dagger, 1\right])                                      & \alpharstar\leq \alpha
    \end{cases}
    \label{eq:nash-equilibria-i}
\end{equation}

\noindent\textbf{Case (ii)}, $c_s \geq \csstar$:
\begin{equation}
    (p_r, p_s)\in \begin{cases}
        (p_s, \left[\psind, \prind\right])                                       & \alpha
        \leq \alphardagger                                                                                                                                                            \\
        (\left[\psstar, 1\right], \psstar)\cup (p_s,\left[\psind,\psstar\right]) & \left[\alphardagger \leq \alpha \leq \alphasdagger\right]\land \left[\psstar \leq p^\dagger\right] \\
        (p_s, \left[\psind,p^\dagger\right])                                     & \left[\alphardagger\leq \alpha \leq \alphasdagger\right]\land \left[p^\dagger\leq \psstar\right]   \\
        (\prstar, \left[p^\dagger, 1\right])                                     & \alpharstar\leq \alpha
    \end{cases}
    \label{eq:nash-equilibria-ii}
\end{equation}\hrule

\section{Refining and interpreting equilibria}
\label{sec:refining-equilibria}
In \eqref{eq:nash-equilibria-i} and \eqref{eq:nash-equilibria-ii} of the last section, we found a proliferation of Nash equilibrium price configurations in $\Scal$ that could be observed given the costs and the value of $\alpha$. However, to help us interpret the Nash equilibria, we would like to discern which price configurations are the ``most likely'' to occur.

To this end, we introduce some criteria to distinguish Nash equilibria:
\begin{enumerate}
    \item \textit{Admissibility} - A Nash equilibrium in which any player plays a weakly dominated strategy\footnote{Player $x$'s strategy $s_x$ is weakly dominated by strategy $s_x'$ if (a) for all strategies of other players $s_y$, $\pi_x(s_x', s_y)\geq \pi_x(s_x, s_y)$, and (b) there exists a strategy of the other players $s_y$ such that $\pi_x(s_x', s_y)> \pi_x(s_x, s_y)$ \cite{rasmusen1989games}.} is inadmissible \cite{govindan2016}.
    \item \textit{Relative Pareto optimality} - We prefer Nash equilibria that are Pareto optimal\footnote{Nash equilibrium $X$ is Pareto suboptimal relative to Nash equilibrium $Y$ if (a) all players weakly prefer, and (b) at least one player strictly prefers $X$ over $Y$ (where ``prefer'' means achieving a higher payoff) \cite{rasmusen1989games}.} relative to all other Nash equilibria.\footnote{Note that a relatively Pareto optimal Nash equilibrium is not necessarily on the Pareto frontier of the game, because it need not be Pareto optimal in the set of all strategies, only in the set of Nash equilibria.}
\end{enumerate}
After these two refinements, we will still be left with some intervals of equilibrium prices. However, most of these are inconsequential; in Section \ref{sec:effective-equilibria} we will see that these refinements allow us to identify who will fulfill demand and the price at which they will sell in equilibrium.

\subsection{Admissibility}
The admissibility criterion allows us to rule out any weakly dominated strategies for the retailer and independent seller. For instance, in the left subplot of Figure \ref{fig:nash-equilibrium-payoffs}, all Nash equilibrium prices $p_r'\in [\psind, \prind)$ are weakly dominated by $\prind$ for the retailer. Indeed, for all $p_r'< \prind$, we have $\pi_{r,s}(p_r')>\pi_{r,r}(p_r')$ by construction of $\prind$. Thus, if the seller sets a price $p_s\in (p_r', \prind)$, then
\begin{equation}\pi_{r}(p_r',p_s) = \pi_{r,r}(p_r') < \pi_{r,r}(\prind)=\pi_{r,s}(\prind) < \pi_{r,s}(p_s)  = \pi_{r}(\prind, p_s)\end{equation}
which implies that the retailer would have preferred to set price $\prind$ rather than $p_r'$. Thus, the only admissible Nash equilibrium in the left subplot of Figure \ref{fig:nash-equilibrium-payoffs} is $p=\prind$.

In Appendix \ref{appendix:admissibility}, we more thoroughly consider the retailer and independent seller's strategies, and find the strategies that are not weakly dominated are
\begin{equation}p_r\in \left[\min\left\{\prind, \prstar\right\}, \max\left\{\prind, \prstar\right\}\right];\quad p_s \in \left[\psind, \psstar\right]\label{eq:admissible-prices-body}\end{equation}
The intuition for the independent seller and the retailer is the same -- we can eliminate all prices that are not between their indifference and optimal prices -- however, for the retailer we do not always have that $\prind \leq \prstar$. Given the admissible strategies in \eqref{eq:admissible-prices-body}, we can simply intersect the sets of admissible prices with the Nash equilibrium price configurations we found in \eqref{eq:nash-equilibria-i} and \eqref{eq:nash-equilibria-ii}. For instance, in the right subplot of Figure \ref{fig:nash-equilibrium-payoffs}, we find that the admissible Nash equilibria are
\[\left[\psind, p^\dagger\right]\cap \left[\prstar,\prind\right] \cap \left[\psind, \psstar\right] =  \left[\prstar, p^\dagger\right]\]
We perform these intersections throughout all of parameter space in Appendix~\ref{appendix:admissibility}, and include a complete summary of admissible equilibria in \eqref{eq:admissible-equilibria-shared}, \eqref{eq:admissible-equilibria-i}, and \eqref{eq:admissible-equilibria-ii}.

\subsection{Relative Pareto optimality}
\label{sec:pareto optimality}

The relative Pareto optimality criterion allows us to further eliminate one admissible equilibrium, which occurs in the following case:
\begin{align}&(p_r,p_s)\in \left(\prstar, \left[p^\dagger, \psstar\right]\right) \cup \left(p_s, \left[\psind, p^\dagger\right]\right)\ ; &&\text{if } \alpharstar \leq \alpha\leq \alphasdagger\text{ and } p^\dagger\leq \psstar\nonumber\end{align}
The equilibrium where the retailer fulfills demand at $\prstar$ is Pareto suboptimal relative to the shared-price equilibrium $p^\dagger$ where the independent seller fulfills demand. Indeed, by construction of $p^\dagger$, in both cases the retailer achieves the same payoff of $\pi_{r,r}(\prstar) = \pi_{r,s}(p^\dagger)$ so they are indifferent. However, because $\alpharstar \leq \alpha \leq \alphasdagger$ implies that $\psind \leq p^\dagger$, the independent seller achieves a positive payoff when fulfilling demand at $p^\dagger$ while they achieve a payoff of $0$ when allowing the retailer to fulfill demand at $\prstar$. Thus, they strictly prefer fulfilling demand at $p^\dagger$ rather than allowing the retailer to fulfill demand at $\prstar$.

As a remark, we note that a reasonable argument could be made for preferring the equilibrium where the retailer fulfills demand at $\prstar$. In this case $\prstar< p^\dagger$, so despite the Pareto optimality argument, the customer will be able to buy the product at a lower price if the retailer fulfills demand at $\prstar$ instead of the independent seller fulfilling demand at $p^\dagger$. With this lowest price criterion, $\prstar$ would in fact be the \textit{most likely} equilibrium in this case, as $\prstar < \psind$ as well. Whatever we choose does not substantially affect the results, and the relative Pareto optimality criterion will continue to be sensible in the larger shared-revenue Bertrand game, thus we proceed with this assumption.

\subsection{Equilibrium outcomes}
\label{sec:effective-equilibria}
Even after applying these two refinements, inspecting \eqref{eq:admissible-equilibria-shared}, \eqref{eq:admissible-equilibria-i}, and \eqref{eq:admissible-equilibria-ii} we can still see that we are left with many intervals of prices. However, in order to interpret the real-world implications of the equilibria, it is prudent to consider the equilibrium outcomes: at what price the customer will purchase the product and from whom they will buy. Indeed, letting the prices be $p\in P\subset [0,1]$ and the fulfiller of demand be $x\in \{r,s\}$, we find the following equilibria:

\begin{center} {\bf Equilibrium Outcomes of \hyperref[game:staying-subgame]{Staying Subgame}}\smallskip
    \hrule
\end{center}{\bf Notation}\\
We denote an equilibrium outcome by $(P,x)\subset (0,1)\times \{r,s\}$, indicating that there is an admissible, Pareto optimal Nash equilibrium where player $x$ fulfills demand at any price $p\in P\subset (0,1)$.

\noindent\textbf{Case (i)}, $c_s \leq \csstar$:
\begin{equation}
    (P,x) = \begin{cases}
        (\prind, s)                                                                                           & \alpha
        \leq \alphasstar                                                                                                                                \\
        (\psstar, s)                                                                                          & \alphasstar\leq \alpha \leq \alphaopt   \\
        \left(\left[\max\left\{\prstar, \psind\right\},\min\left\{\psstar, p^\dagger\right\}\right], s\right) & \alphaopt\leq \alpha \leq \alphasdagger \\
        (\prstar, r)                                                                                          & \alphasdagger\leq \alpha
    \end{cases}
    \label{eq:body-effective-equilibria-i}
\end{equation}

\noindent\textbf{Case (ii)}, $c_s \geq \csstar$:
\begin{equation}
    (P,x) =   \begin{cases}
        (\prind,s)                                                                                            & \alpha
        \leq \alphardagger                                                                                                                                  \\
        \left(\left[\max\left\{\prstar, \psind\right\},\min\left\{\psstar, p^\dagger\right\}\right], s\right) & \alphardagger\leq \alpha \leq \alphasdagger \\
        (\prstar, r)                                                                                          & \alphasdagger\leq \alpha
    \end{cases}
    \label{eq:body-effective-equilibria-ii}
\end{equation}\hrule\bigskip
\noindent Viewing the equilibria in the light of \eqref{eq:body-effective-equilibria-i} and \eqref{eq:body-effective-equilibria-ii} is incredibly helpful for intepretation; the remainder of Section \ref{sec:refining-equilibria} is dedicated to analyzing the outcomes rather than the price configurations.

A quick glance at \eqref{eq:body-effective-equilibria-i} and \eqref{eq:body-effective-equilibria-ii} shows that when $\alpha$ is sufficiently high, the independent seller is ``priced out'' of the market and the retailer will fulfill demand at their single agent price as if the independent seller did not exist. When $\alpha$ is small, the independent seller sells at the retailer's indifference price, illustrating that they have to prioritize preventing the retailer from undercutting them over maximizing their own payoff. However, when $\alpha$ is in a certain sweet spot, the independent seller may be able to sell at their optimal price in equilibrium. We will see soon that everybody benefits in this regime -- if the retailer gets to choose the referral fee, they would prefer to choose a referral fee in this sweet spot rather than one that is too high or too low.


\subsection{Fee and cost structure analysis}
\label{sec:param_space_analysis}

Our results in \eqref{eq:body-effective-equilibria-i} and \eqref{eq:body-effective-equilibria-ii} show that given a point in the parameter space $(c_r, c_s, \alpha)$, we can predict which player will fulfill demand and the price(s) at which they will fulfill demand in an admissible, Pareto optimal Nash equilibrium. In this section, we build intuition about how the parameter space is divided into regions based on these possible equilibria and discuss the business implications for the retailer, both in terms of the retailer's payoff and in terms of the effect on the prices. Importantly, we will contrast these metrics with the alternative of not having the revenue sharing option available to outside sellers -- i.e. compare them with the price the retailer would have set and its payoff in the absence of an independent seller.
\begin{figure}
    \begin{center}
        \includegraphics[height=5.5cm]{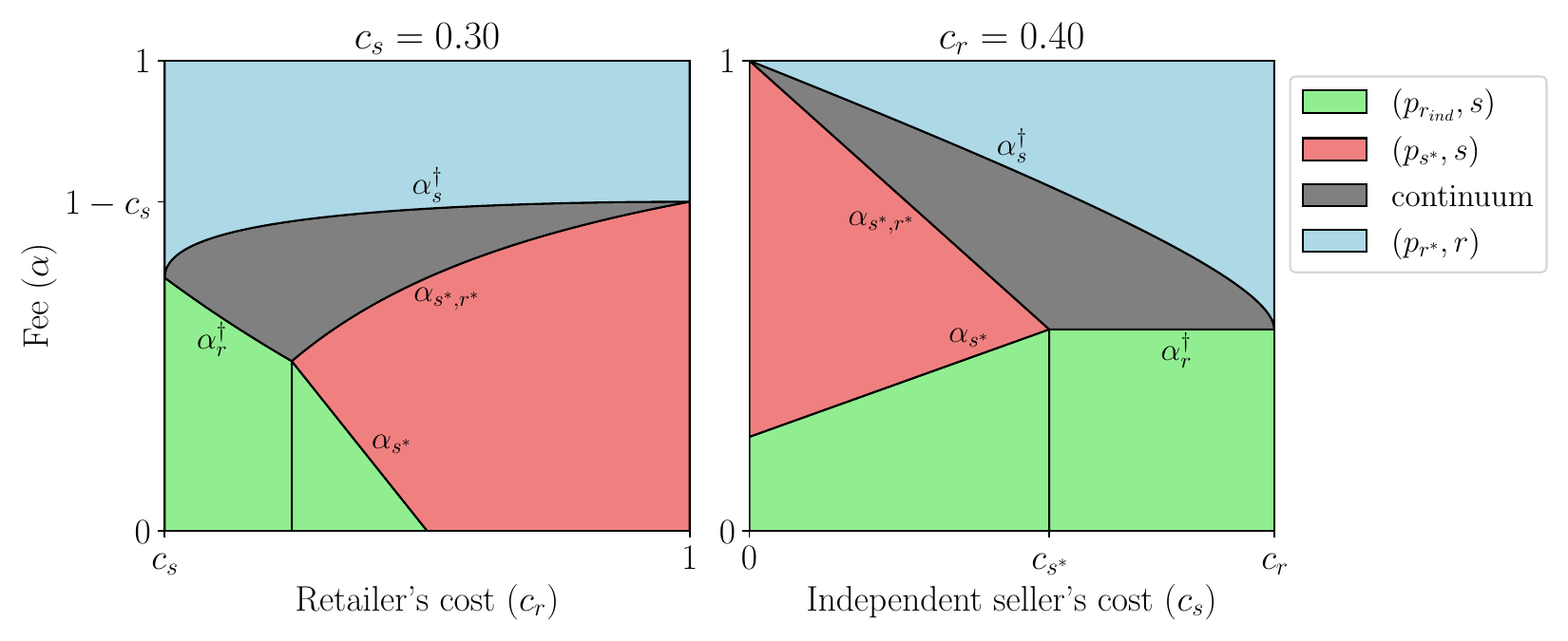}
        \caption{Equilibrium outcomes as a function of $\alpha$ vs $c_r$ (left) and $\alpha$ vs $c_s$ (right).}
        \label{fig:param_space}
    \end{center}
\end{figure}

In Figure~\ref{fig:param_space} we plot two cross sections of the parameter space, one plotting $c_r$ vs $\alpha$ at a fixed $c_s = 0.4$ and one with $c_s$ vs $\alpha$ at a fixed $c_r = 0.6$. The different colors correspond to the admissible, Pareto optimal equilibrium for each set of parameters, with the ``continuum'' corresponding to the interval of prices $\left(p \in\left[\max\left\{\prstar, \psind\right\},\min\left\{\psstar, p^\dagger\right\}\right], s\right)$.

Notice that when $c_s\geq \csstar$ there is no value of $\alpha$ that yields an equilibrium $(\psstar, s)$ in the red region -- this is where the $\csstar$ notation comes from. Additionally, notice that we may have that $\alphasstar< 0$, in which case the green region is not feasible as we cannot have $\alpha \leq \alphasstar< 0$.

\subsubsection{$(\pastar, r)$ - The retailer fulfills demand at its single agent price}

When $\alpha > \alphasdagger$ (blue region), the fee is too high for the independent seller to fulfill demand, and the game reduces to the retailer ``playing'' alone and charging its single agent price. This is the least interesting equilibrium if we wish to consider the retailer's interaction with the independent sellers, but is interesting to be contrasted with all other equilibria as it is the only equilibrium at which the retailer fulfills demand.

\subsubsection{A continuum of equilibria in which the independent seller fulfills demand at a price higher than the retailer's single agent price}
\label{sec:continuum-analysis}

When $\alpha \in \left[\max\left\{\alphaopt, \alphardagger\right\}, \alphasdagger\right]$, there are a continuum of equilibrium prices at which the independent seller could fulfill demand (gray region). Note that the independent seller's price is at least $\prstar$ in this region of parameter space, thus the price to the end customer increases relative to a world in which the revenue sharing program did not exist. However, in this region, the retailer achieves at least as much payoff as they would achieve fulfilling demand on their own without the revenue sharing program. Intuitively, this equilibrium is ``safe'' for the retailer because they will never make less than their single agent payoff, but it comes at the expense of potentially charging customers a higher price.

\subsubsection{$(\paind, s)$ - The independent seller fulfills demand at the retailer's indifference price, $\paind$}

When $\alpha \leq \min\left\{\alphasstar, \alphardagger\right\}$, the independent seller fulfills demand at $p = \paind$ (green region). In this region there is a trade-off between our two metrics -- the retailer's payoff is lower, but the price for the customer is also lower compared to the retailer selling in a single agent game. Furthermore, the existence of this equilibrium implies that the retailer's option of selling directly to customer protects the customer from the independent seller charging too high a price when the fee is low.

\subsubsection{$(\pfstar, s)$ - The independent seller fulfills demand at their single agent price, $\pfstar$}
\label{sec:psstar-analysis}

When $\alpha \in \left[\alphasstar, \alphaopt\right]$, we have an equilibrium in which the independent seller fulfills demand at its optimal price $\pfstar$ (red region). In this regime, we always have $\psstar\leq \prstar$, thus the customers are able to purchase at a lower price than without the revenue sharing program. The independent seller prefers this equilibrium because they are able to sell at their optimal price. Finally, there exist some $\alpha$ in the red region such that the retailer achieves a payoff higher than their optimal single agent payoff, a result which benefits the customers, independent seller, and retailer.

\section{Optimizing the referral fee}
\label{sec:alpha-game-setup}
In Sections \ref{sec:continuum-analysis} and \ref{sec:psstar-analysis}, we identified two equilibrium regions where the retailer's payoff may increase relative to their optimal single agent payoff (without the revenue sharking program). In this section, we allow the retailer to choose the referral fee before playing the staying subgame, which is equivalent to determining which equilibrium is most preferable to the retailer for a given cost configuration. Precisely, we study the following sequential game:

\label{game:fee-optimization}
\begin{center} {\bf Fee Optimization Game}\smallskip
    \hrule
\end{center}{\bf Parameters}\\
Demand curve $q(p)$, retailer's cost $c_r$, seller's cost $c_s$

\noindent {\bf Gameplay}\begin{enumerate}
    \item The retailer chooses the referral fee $\alpha \in (0, 1)$
    \item The retailer and the independent seller play the \hyperref[game:staying-subgame]{staying subgame} with parameters $(q, \alpha, c_r, c_s)$
\end{enumerate}\hrule\bigskip
\noindent Examining Figure \ref{fig:param_space}, we can think of fixing a cost on the horizontal axis and taking a slice, for instance drawing a vertical line corresponding to $c_r = 0.9$ on the left plot. With this, the retailer now gets to choose whatever value of $\alpha$ they'd like, which is a point along the vertical line $c_r = 0.9$, effectively getting to choose which type of equilibrium they would like to observe. We illustrate such a slice in Figure \ref{fig:payoff_example_alpha} by coloring the background of the plot according to which equilibrium in Figure \ref{fig:param_space} is observed.

\subsection{Refined equilibria of fee optimization game}

We will impose the same refinements on the Nash equilibria of the sequential game as we did in Section \ref{sec:refining-equilibria} -- namely, that they must be admissible and Pareto optimal relative to other equilibria. Admissibility has no effect on the equilibria of the sequential game, and Pareto optimality yields us the natural refinement that the retailer will choose the minimum $\alpha$ necessary to attain their maximum possible equilibrium payoff. More formal justification can be found in Appendix \ref{appendix:sequential_game_refinements}.

We will impose the further refinement of subgame perfection \cite{rasmusen1989games}, which is equivalent to backward induction in this game of perfect, complete information. In order to achieve this, the retailer and independent seller must choose a price configuration that is an admissible, Pareto optimal equilibrium of the staying subgame. However, there are many choices of such price configurations. We will thus define an equilibrium strategy profile $\rho$ for the subgame. Formally, we require that $\rho$ is a strategy profile with the property that for every possible $\alpha, c_r, c_s$, $\rho(\alpha, c_r, c_s) \in \Scal$ is an admissible, Pareto optimal equilibrium price configuration of the staying subgame as outlined in \eqref{eq:admissible-equilibria-shared}, \eqref{eq:admissible-equilibria-i}, and \eqref{eq:admissible-equilibria-ii}\footnote{This function $\rho$ is necessary for subgame perfection and could come from anywhere. Intuitively, we could think of the retailer and independent seller agreeing in advance on what prices they will set for each choice of $\alpha$, given common knowledge of their prices $c_r, c_s$. We could also think of the retailer as more powerful and give them the ability to choose $\rho$, requiring that the independent seller agree to choose the prices outlined in $\rho$ as a condition of the revenue sharing program.}.

To give some concrete examples, some natural choices for $\rho$ are:
\begin{enumerate}[(I)]
    \item $\underline{\rho}(\alpha)$ -- for each $\alpha$, $\underline{\rho}(\alpha)\in \Scal$ is an admissible, Pareto optimal price configuration where the equilibrium outcome has demand being fulfilled at the \emph{minimum} possible price
    \item $\overline{\rho}(\alpha)$ -- for each $\alpha$, $\overline{\rho}(\alpha)\in \Scal$ is an admissible, Pareto optimal price configuration where the equilibrium outcome has demand being fulfilled at the \emph{maximum} possible price
\end{enumerate}
By inspection of \eqref{eq:body-effective-equilibria-i} and \eqref{eq:body-effective-equilibria-ii}, the only region in which $\rho$ makes a nontrivial choice is in the gray continuum region of Figure \ref{fig:param_space}, which allows demand to be fulfilled at $p\in \left[\max\left\{\prstar, \psind\right\}, \min\left\{\psstar, p^\dagger\right\}\right]$ when $\alpha \in \left[\max\left\{\alphardagger, \alphaopt\right\}, \alphasdagger\right]$. Continuing our examples, we would have that
\begin{enumerate}[(I)]
    \item For all $\alpha\in \left[\max\left\{\alphardagger, \alphaopt\right\}, \alphasdagger\right]$, $\underline{\rho}(\alpha) \equiv \left(\max\left\{\prstar, \psind\right\}, \max\left\{\prstar, \psind\right\}\right)\in \Scal$
    \item For all $\alpha\in \left[\max\left\{\alphardagger, \alphaopt\right\}, \alphasdagger\right]$, $\overline{\rho}(\alpha) \equiv \left(\min\left\{\psstar, p^\dagger\right\},\min\left\{\psstar, p^\dagger\right\}\right)\in \Scal $
\end{enumerate}
For all $\alpha\notin \left[\max\left\{\alphardagger, \alphaopt\right\}, \alphasdagger\right]$, there is only one admissible, Pareto optimal equilibrium outcome as specified by  \eqref{eq:body-effective-equilibria-i} and \eqref{eq:body-effective-equilibria-ii}, thus we will always observe the same outcome irrespective of our choice of $\rho$. For any $\rho$, the subgame perfect, admissible, Pareto optimal equilibria are as follows:

\begin{center} {\bf Equilibrium of \hyperref[game:fee-optimization]{Fee Optimization Game}}\smallskip
    \hrule
\end{center}{\bf Given}\\
$\rho(\alpha, c_r, c_s)\in \Scal$ -- an admissible, Pareto optimal equilibrium strategy profile of the \hyperref[game:staying-subgame]{staying subgame} for every set of parameters

\noindent {\bf Equilibrium}
\begin{enumerate}
    \item The retailer sets $\alphastar(c_r, c_s, \rho)$, the minimum referral fee that maximizes equilibrium payoff
    \item The retailer and independent seller set the equilibrium price configuration $\rho(\alphastar, c_r, c_s)$
\end{enumerate}\hrule
\subsection{Interpreting equilibria}
\label{sec:interpreting-for-overline-underline-rho}
To complete our analysis, our natural next step is to find $\alphastar$. We'll denote as $\pireq(\rho(\alpha), \alpha, c_r, c_s)$ the payoff function in equilibrium of the staying subgame for the retailer with strategy $\rho$ and $\piseq$ the analogous function for the independent seller. With this notation, it's clear that
\begin{equation}\alphastar(c_r, c_s, \rho) \equiv \min_\alpha\left[\arg\max_\alpha \pireq(\rho(\alpha), \alpha, c_r, c_s)\right]\footnotemark\label{eq:alpha-star-definition}\end{equation}
\footnotetext{Formally, we should use $\inf$ and $\sup$ in \eqref{eq:alpha-star-definition}; we can easily construct a $\rho$ such that uncountably many $\alpha$ maximize $\pireq(\rho)$.}However, the functional dependence of $\alphastar$ on $\rho$ makes this optimization problem complicated in general. To build intuition, we will begin by considering only $\underline{\rho}$ and $\overline{\rho}$. In Appendix \ref{appendix:upper-bounding-retailer-payoff}, we show the simple result that
\[\alphastar(c_r,c_s,\underline{\rho}) = \alpharstar \]
In other words, if the retailer and independent seller agree to play strategy $\underline{\rho}$, the retailer will always choose referral fee $\alpharstar$ in equilibrium of the \hyperref[game:fee-optimization]{fee optimization game}.

The optimal fee for $\overline{\rho}$ is slightly more complicated than for $\underline{\rho}$, but we show in Appendix \ref{appendix:lower-bounding-retailer-payoff} that
\begin{equation}\alphastar(c_r,c_s,\overline{\rho}) = \begin{cases}
        \min_\alpha[\arg\max_\alpha \pi_{r,s}(\psstar, \alpha)]\equiv \overline{\alpha} & \exists \alpha \text{ s.t. } \pi_{r,s}(\psstar, \alpha)> \pi_{r,r}(\prstar)\footnotemark \\
        \alphardagger                                                                   & \text{else}
    \end{cases}\label{eq:alphastar-cases-underline-rho}\end{equation}\footnotetext{We show in Lemma \ref{lemma:csstar-implies} that $c_s \leq \csstar\implies \alphastar(\overline{\rho}) = \overline{\alpha}$.}\begin{figure}
    \begin{center}
        \includegraphics[height=5.5cm]{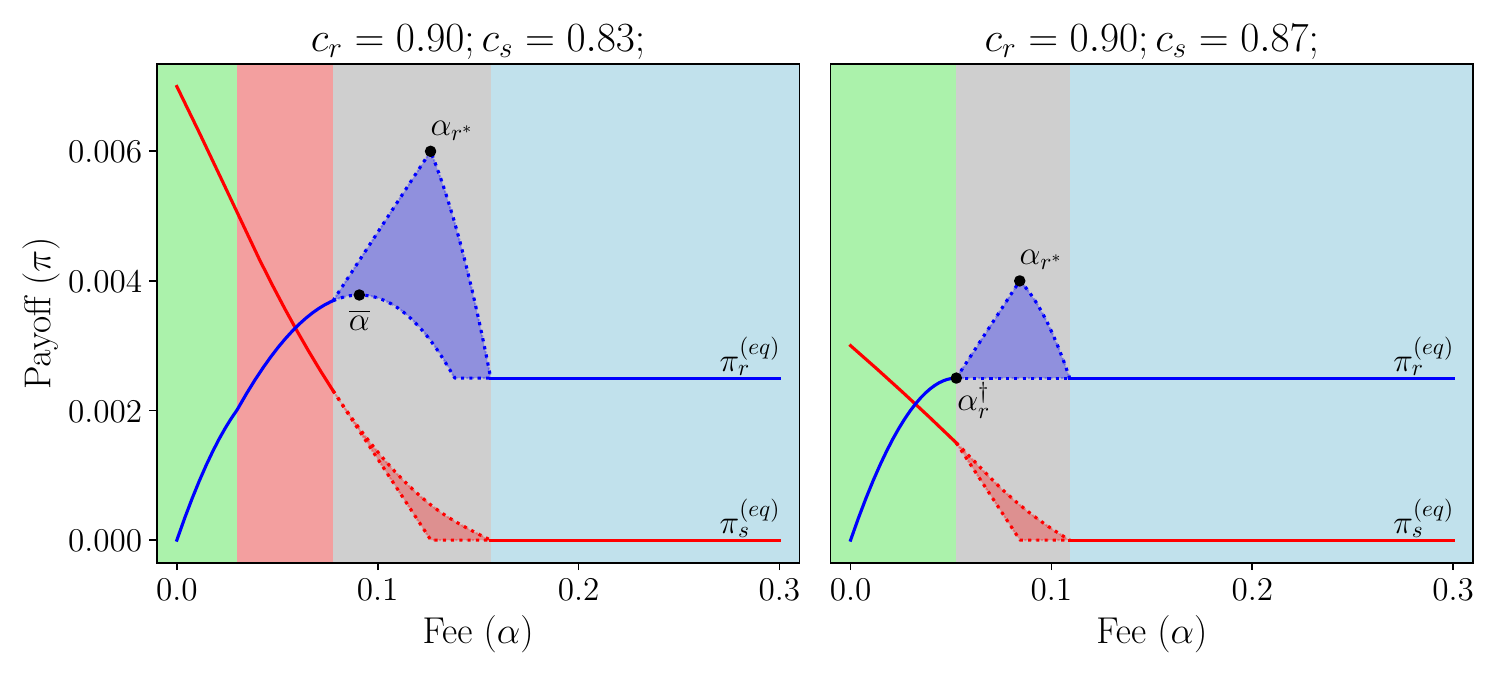}
        \caption{Equilibrium payoff curves}
        \label{fig:payoff_example_alpha}
    \end{center}
\end{figure}To illustrate $\alphastar$, we plot typical payoff curves in Figure \ref{fig:payoff_example_alpha}. The background of the plot corresponds to the equilibria in Figure \ref{fig:param_space}. In the left plot, there does exist an $\alpha$ such that $\pi_{r,s}(\psstar, \alpha)> \pi_{r,r}(\prstar)$, so $\alphastar(\overline{\rho})=\overline{\alpha}$. In the right plot, for all $\alpha$ we have $\pi_{r,s}(\psstar, \alpha)\leq \pi_{r,r}(\prstar)$, so  $\alphastar(\overline{\rho})=\alphardagger$.
\subsection{General equilibrium strategy profiles}

In Section \ref{sec:interpreting-for-overline-underline-rho}, we interpreted the equilibria of the \hyperref[game:fee-optimization]{fee optimization game} under specifically the equilibrium strategy profiles $\underline{\rho}$ and $\overline{\rho}$. However, as a result of monotonicity of the payoff functions in the gray region (the only region where the choice of $\rho$ is nontrivial), we can bound the payoffs for any $\rho$ by the payoffs for $\underline{\rho}$ and $\overline{\rho}$. Thus, the shaded areas in the gray region of Figure \ref{fig:payoff_example_alpha} correspond to the range of possible payoffs for the retailer and seller, with the exact values depending on the choice of $\rho$. Slightly more formally, in Appendix \ref{appendix:equilibrium-strategy} we prove the following:
\setcounter{corollary}{2}
\begin{corollary}
    For any strategy profile $\rho$, the retailer's (seller's) equilibrium payoff is upper bounded (lower bounded) by their equilibrium payoff under $\underline{\rho}$, and the retailer's (seller's) equilibrium payoff is lower bounded (upper bounded) by their equilibrium payoff under $\overline{\rho}$.
    \label{cor:bounded-equilibrium-payoffs}
\end{corollary}
In light of Corollary \ref{cor:bounded-equilibrium-payoffs}, we can bound the equilibrium payoff of the fee optimization game $\pi_r$ for any strategy profile $\rho$ --
\[\max_\alpha \pireq(\overline{\rho}(\alpha), \alpha, c_r, c_s) \leq \pireq(\rho(\alphastar), \alphastar, c_r, c_s) = \pi_r(\rho(\alpha),c_r,c_s) \leq \max_\alpha \pireq(\underline{\rho}(\alpha), \alpha, c_r, c_s)\]
Bounds on $\pi_s$ and $\alphastar$ follow similarly by considering $\pireq$ under the strategy profiles $\underline{\rho}$ and $\overline{\rho}$. Indeed, we show in Appendix \ref{appendix:lower-bounding-retailer-payoff} and \ref{appendix:upper-bounding-retailer-payoff} that for any equilibrium strategy profile $\rho$, the admissible, Pareto optimal, subgame perfect Nash equilibria of the sequential game satisfy:
\begin{equation}\pi_r \in  \begin{cases}
        \left[\pi_{r,s}(\psstar, \overline{\alpha}), \pi_{r,s}\left(\prstar, \alpharstar\right)\right] & \exists \alpha \text{ s.t. } \pi_{r,s}(\psstar, \alpha)> \pi_{r,r}(\prstar) \\
        \left[\pi_{r,r}(\prstar), \pi_{r,s}\left(\prstar, \alpharstar\right)\right]                    & \text{else}
    \end{cases}
    \label{eq:retailer-payoff-bounds-sequential-game}
\end{equation}
\begin{equation}
    \pi_s \in \left[0, \pi_{s,s}\left(\psstar, \min\left\{\alphasstar, \alphardagger\right\}\right)\right]\ ;\qquad \alphastar \in \left[\min\left\{\alphasstar, \alphardagger\right\}, \alphasdagger\right]
    \label{eq:seller-payoff-bounds-sequential-game}
\end{equation}

One important takeaway from \eqref{eq:retailer-payoff-bounds-sequential-game} and \eqref{eq:seller-payoff-bounds-sequential-game} is that the equilibrium $\alphastar$ is bounded away from $1$. The existence of such an optimal fee is echoed by optimal taxation theories such as the Laffer curve \cite{laffer2004} -- it is not in the retailer's interest to take all the independent seller's revenue.

Importantly, examining Figure \ref{fig:payoff_example_alpha}, we see that the retailer will always induce an equilibrium where the independent seller fulfills demand at some $p\in\left[\max\left\{\psind, \prstar\right\}, \left\{\psstar, p^\dagger\right\}\right]$ depending on $\rho$ (red or gray regions). Interestingly, the socially suboptimal green region where the independent seller fulfills demand at $\prind$ is \emph{never} observed in equilibrium of the fee optimization game. The retailer's payoff always increases relative to a world where the revenue sharing program did not exist.

We do also find that the retailer never sets the fee high enough such that they fulfill demand themselves at price $\prstar$ in equilibrium (blue region). However, this is less interesting, because it is an artifact of the relative Pareto optimality criterion that we applied when refining equilibria in Section \ref{sec:pareto optimality}; we argued there that the retailer would prefer for the independent seller to fulfill demand at $p^\dagger$ rather fulfilling demand themselves at $\prstar$.
\section{An outside option for the independent seller}
\label{sec:outside-option}
In the \hyperref[game:fee-optimization]{fee optimization game}, the retailer was free to choose any fee they'd like, and the independent seller would stay in the revenue sharing program irrespective of the fee chosen. To complete our analysis of the \hyperref[game:shared-revenue-bertrand]{shared-revenue Bertrand game}, we now model the independent seller's choice to leave the revenue sharing program as a single-step negotiation with the retailer. Indeed, in real scenarios the independent seller often has a compelling alternative (outside option) -- selling their product independently rather than participating in the revenue sharing program.

\subsection{Leaving subgame}
\label{sec:leaving-subgame}
To finish defining the \hyperref[game:shared-revenue-bertrand]{shared-revenue Bertrand game}, we must outline the structure of the leaving subgame. When not clear from context, we will use superscripts $\cdot^{(\ell)}$ to denote relevant quantities in the leaving subgame and $\cdot^{(o)}$ to denote relevant quantities in the shared-revenue Bertrand game with the outside option. There are many reasonable specifications for this leaving subgame \cite{shubik1980market,tirole1988theory}, and more sophisticated models may better capture reality, but as a first step we choose a very simple model: if the independent seller chooses to leave, they will play a standard Bertrand game against the retailer \cite{bertrand1883review}.

With the constraint that $\alpha > 0$, it seems obvious that the independent seller would always leave to avoid paying the referral fee. Thus, we will posit that if they choose to sell on their own, the effective cost to the independent seller will be $c_s + \delta$. This additive $\delta$ factor could represent any number of costs that the independent seller might incur if selling without the retailer such as advertising, shipping, or storing inventory. It could even represent an opportunity cost that the smaller independent seller incurs by not exposing their product to the larger retailer's customers.

Whatever it represents, we will assume that there exists a $\delta > 0$ such that the independent seller does not always prefer to leave. We now have a standard Bertrand game with potentially asymmetric costs depending on the value of $\delta$, for which the equilibria have been derived in \cite{blume2003,kartik2011} for instance. However, if $\delta \geq c_r - c_s$ such that $c_s + \delta\geq c_r$, the independent seller achieves an equilibrium payoff of $0$ if they exercise their outside option, making leaving the revenue sharing program a non-credible threat. This can be seen by considering the equilibrium in the regime where $c_r \leq c_s + \delta$ instead of $c_s + \delta < c_r$. Thus, with the constraint that $\delta \in (0, c_r-c_s)$, the admissible equilibrium outcomes of the leaving subgame are
\[(p,x) = \left(\min \left\{c_r, \psstarell(\delta)\right\}, s\right)\]
where $\psstarell(\delta)$ is the maximizer of the independent seller's payoff function in the leaving subgame. Importantly, because the retailer never fulfills demand in equilibrium and there is no revenue sharing, they always achieve an equilibrium payoff of $0$ if the independent seller chooses to leave. Because the retailer achieves nonzero payoff in equilibrium otherwise, they would always prefer to choose an $\alpha$ such that the independent seller stays over an $\alpha$ such that the independent seller leaves. Furthermore, notice that the leaving subgame itself is entirely independent of $\alpha$.

\subsection{Refined equilibria of shared-revenue Bertand game}

The analysis in Appendix \ref{appendix:sequential_game_refinements} about refinements for the \hyperref[game:fee-optimization]{fee optimization game} equilibria applies to the \hyperref[game:shared-revenue-bertrand]{shared-revenue Bertrand game} with outside option: a strategy of the sequential game is admissible if $\rho$ is admissible, and Pareto optimality implies that the retailer will choose the lowest fee possible to achieve their maximum payoff. Furthermore, because the retailer achieves nonzero payoff if the independent seller chooses to stay and $0$ payoff if they leave, Pareto optimality also tells us that if the independent seller is indifferent between staying and leaving, they will choose to stay. To summarize, the admissible, Pareto optimal, subgame perfect equilibrium strategies are:

\begin{center} {\bf Equilibrium of \hyperref[game:shared-revenue-bertrand]{Shared-Revenue Bertrand Game}}\smallskip
    \hrule
\end{center}{\bf Given}\\
$\rho(\alpha, c_r, c_s)\in \Scal$ -- an admissible, Pareto optimal equilibrium strategy profile of the \hyperref[game:staying-subgame]{staying subgame} for every set of parameters

\noindent {\bf Equilibrium}
\begin{enumerate}
    \item The retailer sets $\alphastarell(\delta, c_r, c_s, \rho)$, the minimum referral fee that maximizes equilibrium payoff and does not cause the independent seller to strictly prefer leaving
    \item The independent seller stays in the revenue sharing program
    \item The retailer and independent seller play the \hyperref[game:shared-revenue-bertrand]{shared-revenue Bertrand game} and simultaneously set the equilibrium price $\rho(\alphastarell, c_r, c_s)$
\end{enumerate} \hrule\bigskip
\noindent To be precise, $\alphastarell$ must satisfy the following property similar to \eqref{eq:alpha-star-definition}:
\begin{equation}
    \alphastarell(\delta, c_r, c_s, \rho) \equiv \min_\alpha \left[\arg \max_\alpha\left\{\pireq(\alpha, \rho) \mid \piseq(\alpha, \rho)\geq \pisell(\delta)\right\}\right]
    \label{eq:alphastarell-general}
\end{equation}
\begin{figure}
    \begin{center}
        \includegraphics[height=5.5cm]{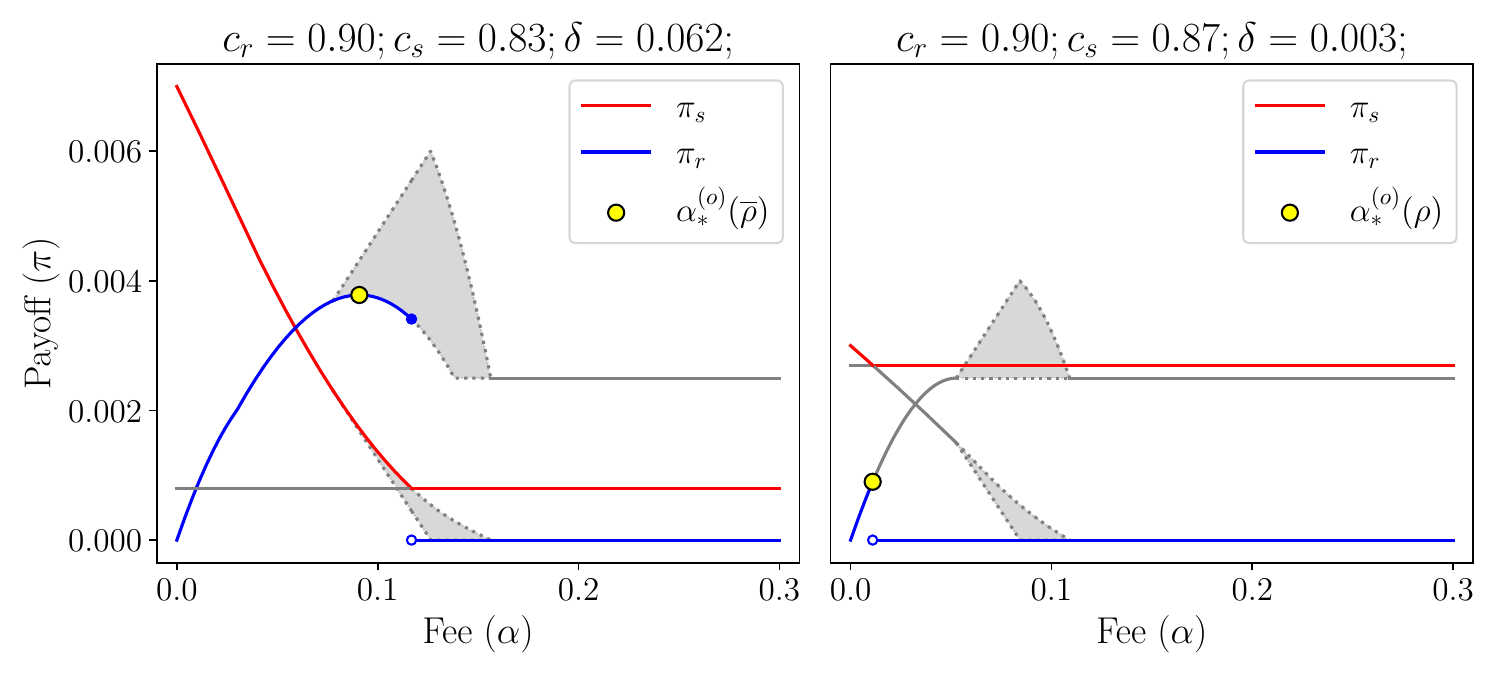}
        \caption{Equilibrium payoff curves with an outside option}
        \label{fig:payoff_example_delta}
    \end{center}
\end{figure}As before, because of the functional dependence of $\alphastarell$ on $\rho$, this is complicated in general. Just as in Section \ref{sec:interpreting-for-overline-underline-rho}, we'll build intuition by analyzing simpler cases first. In Figure \ref{fig:payoff_example_delta}, we plot the same equilibrium payoff curves as Figure \ref{fig:payoff_example_alpha}, however we now include an outside option with $\delta = 0.062$ in the left subplot and $\delta = 0.003$ in the right subplot. For simplicity, in the left subplot we assume that $\rho=\overline{\rho}$. If the retailer sets the fee above some threshold, the independent seller would prefer to leave the revenue sharing program and sell on their own instead. It will be helpful to give this threshold a name; let's define
\begin{equation}\alphamax(\delta,c_r,c_s,\rho) \equiv \sup \left\{\alpha \mid \piseq(\alpha, \rho)\geq \pisell(\delta)\right\}\label{eq:alphamax-definition}\end{equation}
If the retailer chooses a fee $\alpha >\alphamax$, they will achieve a payoff of $0$ in equilibrium of the \hyperref[game:shared-revenue-bertrand]{shared-revenue Bertrand game} because the independent seller will leave. The independent seller, on the other hand, will achieve a payoff of $\pisell(\delta)$ for $\alpha> \alphamax$, which is therefore a lower bound on their equilibrium payoff. Additionally, in equilibrium of the \hyperref[game:shared-revenue-bertrand]{shared-revenue Bertrand game} the retailer will choose a fee $\alpha\leq \alphamax$, which bounds the maximum fee further away from $1$.

Interestingly, the left subplot of Figure \ref{fig:payoff_example_delta} demonstrates that the retailer will not necessarily choose $\alphastarell = \alphamax$, indeed their payoff under $\overline{\rho}$ is maximized by choosing $\overline{\alpha}\in (0, \alphamax)$. However, if we simply switch the equilibrium strategy profile from $\overline{\rho}$ to $\underline{\rho}$, the retailer will choose $\alphastarell = \alphamax$, which serves as a good reminder that the choice of $\rho$ may have a significant effect on the equilibrium outcome in general.

The right subplot of Figure \ref{fig:payoff_example_delta} shows another interesting case where $\delta$ is low enough such that the independent seller achieves a higher payoff than the retailer in equilibrium. In fact, the independent seller is so competitive with the retailer that the retailer achieves a significantly lower payoff in equilibrium of the \hyperref[game:shared-revenue-bertrand]{shared-revenue Bertrand game} than they would if the independent seller did not exist. Additionally, notice that $\alphamax< \alphardagger$ is not in the continuum region. Because of this, the choice of $\rho$ does not matter, which should be clear graphically because the entire continuum region is ``grayed out'' for both the retailer and independent seller. In fact, we can say more generally that
\begin{equation}\alphastarell = \min\left\{\alphastar, \alphamax\right\}\quad \text{ if } \alphamax \leq \min\left\{\alphardagger, \alphasstar\right\}\end{equation}
which is a simple special case of \eqref{eq:alphastarell-general}.


\section{Conclusion and future work}

In this work, we presented a game theory model of a shared-revenue Bertrand game to analyze the dynamics and incentives induced by revenue sharing programs. We model the revenue sharing program as a negotiation where the retailer chooses the referral fee then the independent seller chooses whether to participate or leave and sell on their own. Our findings resonate strongly with existing results from duopolistic competition and optimal taxation theory. In equilibrium, both players are able to achieve positive payoff, and the referral fee is set low enough to stimulate economic activity while high enough to maximize payoff.

However, this work is just a starting point, and there are still many future research directions in which this model can be extended with even more potential implications for understanding the dynamics of revenue sharing:
\begin{itemize}
    \item We have assumed throughout that this is a game of perfect, complete information -- modeling uncertainty or forecasting of the costs or demand curves are interesting next steps.
    \item We do not closely consider the manufacturing of the product, but we would expect the dynamics to differ in the cases that (a) there is one primary manufacturer of the product, from whom both the retailer and indepedent seller buy, or (b) there are many manufacturers of the product, in which case the margins are low.
    \item We could impose capacity constraints on the players, allowing us to reason about strategies for stocking both retailers' and independent sellers' products in the same warehouse.
    \item We look only at a duopolistic interaction, but more sophisticated models of entry and exit from a shared-revenue Bertrand marketplace may help make the game more realistic.
    \item A large retailer or ridesharing service may not want to choose a different fee for each individual independent seller on its platform, instead preferring to choose an ``aggregate fee'' for a homogeneous group of independent sellers.
\end{itemize}
Ultimately, this work marks an important first step towards understanding cooperative business practices between economic agents from first principles.
\newpage
%
\begin{acks}
    We would like to thank Dirk Bergemann for his insightful comments and suggestions.
\end{acks}

\bibliographystyle{ACM-Reference-Format}
\bibliography{main}
\newpage
\appendix
\section{Notation table}
\label{appendix:notation_table}
\newlength\q
\setlength\q{\dimexpr .65\textwidth  -2\tabcolsep}
\noindent\begin{longtable}{l l p{\q}}
    \toprule
    \textbf{Notation}                            & \textbf{Value(s)}                        & \textbf{Description}                                                                                                                                                                                               \\
    \midrule
    $c_x$                                        & $\in (0,1)$                              & Marginal cost for player $x$                                                                                                                                                                                       \\
    $p_x$                                        & $\in (0,1)$                              & Price set by player $x$                                                                                                                                                                                            \\
    $\alpha$                                     & $\in (0,1)$                              & Referral fee                                                                                                                                                                                                       \\
    $\beta$                                      & $\in [0,1]$                              & Market split proportion if prices are equal                                                                                                                                                                        \\
    $q(p)$                                       & $\in [0,1]$                              & Quantity demanded at price $p$                                                                                                                                                                                     \\
    $\Scal$                                      & $(0,1)\times (0,1)$                      & Space of price configurations, eg. $(p_r, p_s)\in \Scal$                                                                                                                                                           \\
    $(p, x)$                                     & $\in (0,1)\times\{r,s\}$                 & Outcome where player $x$ fulfills demand at price $p$                                                                                                                                                              \\
    $\delta$                                     & $\in (0, c_r- c_s)$                      & Additional marginal cost if independent seller leaves                                                                                                                                                              \\
    $\rho(\alpha, c_r, c_s)$                     & $\in \Scal$                              & Equilibrium strategy profile of \hyperref[game:staying-subgame]{staying subgame}                                                                                                                                   \\
    $\overline{\rho}(\alpha, c_r, c_s)$          & $\in \Scal$                              & Equilibrium strategy profile with maximum prices                                                                                                                                                                   \\
    $\underline{\rho}(\alpha, c_r, c_s)$         & $\in \Scal$                              & Equilibrium strategy profile with minimium prices                                                                                                                                                                  \\
    $\pi_x(\cdot)$                               & $\in [0,1)$                              & Player $x$'s payoff                                                                                                                                                                                                \\
    $\pi_x^{\text{(eq)}}(\rho,\alpha, c_r, c_s)$ & $\in [0,1)$                              & Payoff in equilibrium of staying subgame                                                                                                                                                                           \\
    $\pi_x^{(\ell)}(\delta, c_r, c_s)$           & $\in [0,1)$                              & Payoff in equilibrium of leaving subgame                                                                                                                                                                           \\
    $\pi_{x,y}(p_y, \alpha)$                     & $\in [0,1)$                              & Player $x$'s payoff when $y$ fulfills demand at $p_y$                                                                                                                                                              \\
    $\csstar$                                    & $\frac{c_r^2}{\prstar}$                  & Seller optimal feasibility cost; $c_s\leq \csstar$ is necessary for $(\psstar, s)$ to be a potential equilibrium                                                                                                   \\
    $\psstar$                                    & $\arg\max_p\pi_{s,s}(p)$                 & Seller optimal single agent price, the maximizer of $\pi_{s,s}$                                                                                                                                                    \\
    $\psstar^{(\ell)}$                           &                                          & Seller optimal single agent price in leaving subgame                                                                                                                                                               \\
    $\prstar$                                    & $\arg\max_p\pi_{r,r}(p)$                 & Retailer optimal single agent price, the maximizer of $\pi_{r,r}$                                                                                                                                                  \\
    $\prind$                                     & $\frac{c_r}{1 - \alpha}$                 & Retailer indifference price, satisfying $\pi_{r,r}(\prind) = \pi_{r,s}(\prind)$                                                                                                                                    \\
    $\psind$                                     & $\frac{c_s}{1 - \alpha}$                 & Seller indifference price, which is also their breakeven price as it satisfies $\pi_{s,s}(\psind) = \pi_{s,r}(\psind) =0$                                                                                          \\
    $p^\dagger$                                  &                                          & Retailer optimality equivalence price with $\pi_{r,s}(p^\dagger) = \pi_{r,r}(\prstar)$                                                                                                                             \\
    $\tilde{p}$                                  & $\arg\max_p\pi_{r,s}(p)$                 & Revenue maximizing price                                                                                                                                                                                           \\
    $\alphasdagger$                              &                                          & Seller $p^\dagger$ relevance fee, below which selling at $p^\dagger$ becomes a potential best response                                                                                                             \\
    $\alphardagger$                              & $1-\frac{c_r}{\prstar}$                  & Retailer $p^\dagger$ relevance fee, above which allowing the independent seller to sell at $p^\dagger$ becomes a potential best response                                                                           \\
    $\alphasstar$                                & $1-\frac{c_r}{\psstar}$                  & Seller optimal price feasibility fee, above which $(\psstar, s)$ becomes a potential equilibrium                                                                                                                   \\
    $\alpharstar$                                & $1-\frac{c_s}{\prstar}$                  & Retailer optimal price feasibility fee, above which $(\prstar, r)$ becomes a potential equilibrium                                                                                                                 \\
    $\alphaopt$                                  & $1-\frac{c_s}{c_r}$                      & Optimality switching fee satisfying $\psstar(\alphaopt) = \prstar$                                                                                                                                                 \\
    $\alphastar$                                 & \eqref{eq:alpha-star-definition}         & Retailer equilibrium alpha in the \hyperref[game:fee-optimization]{fee optimization game}                                                                                                                          \\
    $\overline{\alpha}$                          & \eqref{eq:alphastar-cases-underline-rho} & Retailer equilibrium fee with strategy profile $\overline{\rho}$                                                                                                                                                   \\
    $\alphastarell$                              & \eqref{eq:alphastarell-general}          & Retailer equilibrium fee for the \hyperref[game:shared-revenue-bertrand]{shared-revenue Bertrand game}                                                                                                             \\
    $\alphamax$                                  & \eqref{eq:alphamax-definition}           & Maximum fee for which the independent seller prefers to stay in the revenue sharing program rather than exercise their outside option in the \hyperref[game:shared-revenue-bertrand]{shared-revenue Bertrand game} \\
    \bottomrule
    \caption{Notation table}
    \label{tab:notation}
\end{longtable}
\noindent Table \ref{tab:notation} summarizes the symbols used throughout the paper along with their value or range of possible values. We use $r$ to denote the retailer and $s$ to denote the seller, with $x,y\in \{r,s\}$ placeholders in quantities that are defined analogously for both players. When clear from context throughout the paper, we drop some of the explicit arguments of the functions -- note for instance that the prices and fees implicitly depend on the costs.

\section{Justification of game setup}
\label{appendix:justification_of_setup}
\subsection{The retailer and seller can achieve nonzero payoff}
We assume in Section \ref{sec:preliminaries} that $c_s < c_r < 1$. Indeed, suppose on the other hand that $c_r\geq 1$. Then, the game becomes trivial:
\begin{lemma}
    In a staying subgame with $c_r\geq 1$, the only admissible Nash equilibrium is $(p_r = 1, p_s = \min\left\{\psstar, 1\right\})$.
\end{lemma}
\begin{proof}
    The retailer's dominant strategy is setting $p_r = 1$. Consider a price $p_r < 1$. If $p_s\leq p_r$,
    \[\pi_r(p_r, p_s) = \pi_{r,s}(p_s) = \pi_r(1, p_s)\]
    On the other hand, if $p_s \in (p_r,1]$, we have
    \[\pi_r(p_r, p_s) = \pi_{r,r}(p_r) < 0 < \pi_{r,s}(p_s) = \pi_r(1, p_s)\]
    With this, the independent seller's best response is clearly $p_s = \min\left\{1,\psstar\right\}$.
\end{proof}

Enforcing this requirement allows us to avoid such trivial equilibria where one or both players abstain from the market because they will simply never fulfill demand at any price.

\subsection{Independent seller's cost is lower}
We also assume in Section \ref{sec:preliminaries} that $c_s < c_r$. Suppose on the other hand that $c_r \leq c_s$. Then, the game becomes trivial:
\begin{lemma}
    In a staying subgame with $c_r< c_s$, there are no Nash equilibria where the independent seller fulfills demand.
\end{lemma}
\begin{proof}
    Suppose for contradiction there exists an equilibrium $(p_r\geq p_s, p_s)$ where the independent seller fulfills demand. Then, we claim the following must be satisfied:
    \begin{enumerate}[(I)]
        \item $p_s\geq \psind$
        \item $p_s \leq \prind$
    \end{enumerate}
    Indeed, to show (I) suppose for contradiction $p_s \leq p_r < \psind$. Then, we have
    \[\pi_s(p_r,p_s) = \pi_{s,s}(p_s)<0=\pi_{s,r}(p_r) = \pi_s(p_r, \psind)\]
    Similarly, if $p_s < \psind\leq p_r$, we have $\pi_s(p_r, \psind) = \pi_{s,s}(\psind) = 0$.

    To show (II), suppose for contradiction $\prind<p_s\leq p_r$. Then, by construction of $\prind$, we have
    \[\pi_r(p_r, p_s) = \pi_{r,s}(p_s) < \pi_{r,r}(p_s)\]
    by continuity of $\pi_{r,r}$, this implies that there exists some $\varepsilon>0$ such that $\pi_{r,r}(p_s - \varepsilon)> \pi_{r,s}(p_s)$, and therefore $p_r$ yields a lower payoff than $p_s-\varepsilon$ so the retailer is not best responding.

    However, both (I) and (II) cannot be simultaneously satisfied, because under the assumption that $c_r<c_s$ we have that
    \[\prind = \frac{c_r}{1-\alpha}< \frac{c_s}{1-\alpha} = \psind\]
    Thus, such an equilibrium cannot exist.
\end{proof}
Because we are interested in modeling the dynamics between the retailer and independent seller, if there are no equilibria where the independent seller fulfills demand our model is entirely uninteresting. Intuitively, if the retailer can fulfill demand themselves for less cost per unit, they have no incentive to initiate a revenue sharing program in the first place.

This implies that the revenue sharing program is best thought of as permitting the sale of niche items or services that the retailer does not specialize in. Instead of specializing in this offering themselves, the model aims to find the conditions under which it would be beneficial for the retailer to partner with the specialist independent seller.

\section{Key prices and lemmas}
\label{appendix:key-prices}

\subsection{No split market equilibria}
\label{appendix:no_split_market_proof}
In this section, we present the full proof of Proposition \ref{thm:no_split_market}, reprised here for convenience --
\setcounter{proposition}{0}
\begin{proposition}
    \label{thm:no_split_market_reprised}
    Let $(p_r, p_s) = (p, p)$, with the equilibrium quantities sold by each player $q_r = \beta q(p)$ and $q_s = (1 - \beta) q(p)$. Then $\beta \in \{ 0, 1 \}$.
\end{proposition}

\begin{proof}
    Suppose for contradiction $\beta\in (0,1)$. For such an equilibrium, the retailer's payoff is
    \[ \pi_r(p) = \beta\pi_{r,r}(p) + (1-\beta)\pi_{r,s}(p)\]
    If $p<\prind$, then the retailer clearly prefers to increase their price to $p + \varepsilon$ and allow the independent seller to fulfill demand at price $p$, achieving payoff $\pi_{r,s}(p)$ rather than the weighted average.

    If $p>\prind$, then $\pi_{r,r}(p)>\pi_{r,s}(p)$, so by continuity of $q(p)$ there exists some $\varepsilon>0$ such that the retailer prefers to achieve payoff $\pi_{r}(p-\varepsilon)=\pi_{r,r}(p-\varepsilon)$. Specifically, suppose that $\pi_{r,r}(p-\varepsilon)=\pi_{r,r}(p)-\delta$. Indeed, we need to find a $\delta$ satisfying
    \[\pi_{r,r}(p)-\delta = \pi_{r,r}(p-\varepsilon)> \beta\pi_{r,r}(p) + (1-\beta)\pi_{r,s}(p)\iff \delta< (1-\beta)(\pi_{r,r}(p)-\pi_{r,s}(p))\]
    which always exists by continuity of $\pi_{r,r}(p)$ and the fact that $\pi_{r,r}(p)-\pi_{r,s}(p)>0$.

    Finally, if $p=\prind$, we have $p>\psind$, so a similar argument as the previous case for the independent seller implies that they would prefer to increase their price to $p-\varepsilon$, allowing them to achieve a payoff of
    \[\pi_{r,s}(p-\varepsilon)>(1-\beta)\pi_{r,s}(p) \]
    by continuity of $\pi_{r,s}(p)$. Thus, if $\beta \notin \{ 0, 1 \}$, the proposed price and market split does not constitute an equilibrium.
\end{proof}

\subsection{Payoff function properties}

First, note that the revenue function $pq(p)$ is strictly concave because
\[\frac{\partial^2}{\partial p^2}pq(p) =\frac{\partial}{\partial p}(q(p)+pq'(p)) = (1+p)q''(p) + q'(p) <0\]
thus, it is uniquely maximized on $(0,1)$ and the price $\tilde{p}\equiv \arg\max_p pq(p)$ is well-defined. Notice $\tilde{p}$ is trivially also the maximizer of $\pi_{r,s}(p) = (1-\alpha)pq(p)$. For the payoff function $\pi_{r,r}(p) = pq(p)-c_rq(p)$, using the above result we have
\[\frac{\partial^2}{\partial p^2}\pi_{r,r}(p) = \frac{\partial^2}{\partial p^2}(pq(p)-c_rq(p)) = (1-c_r+p)q''(p)+q'(p)\]
and since $c_r<1$, $\pi_{r,r}$ is also strictly concave and uniquely maximized which implies $\prstar$ is well-defined. Similarly, $\pi_{s,s}(p)$ is strictly concave for all prices whenever $c_s\leq1-\alpha$, so $\psstar$ is well-defined. For $c_s>1-\alpha$, there is no price at which the independent seller can achieve positive payoff, so $\psstar$ will not be relevant. We now remark about the relative ordering of these maximizers --
\begin{lemma}
    Let $\tilde{p}\equiv \arg\max_p\pi_{r,s}(p)$. Then, $\prstar, \psstar, p^\dagger>\tilde{p}$, and $\psstar>\psind$ for $\alpha<1-c_s$.
    \label{lemma:ptilde-ordering}
\end{lemma}
\begin{proof}
    Consider the first order conditions satisfied by $\tilde{p}$ and $\prstar$ --
    \begin{align}
        \frac{q'(\tilde{p})}{q(\tilde{p})}
         & = -\frac{1}{\tilde{p}}
        \label{eq:ptilde-foc}        \\
        \frac{q'(\prstar)}{q(\prstar)}
         & = \frac{1}{c_r - \prstar}
        \label{eq:prstar-foc}
    \end{align}
    Notice that the left side is the same in both equations, while the function on the right side is shifted by $-c_r$. Thus, the price at which the two functions intersect also shifts right, implying that $\prstar>\tilde{p}$. Now, consider the first order condition for $\psstar$ --
    \begin{equation}
        \frac{q'(\psstar)}{q(\psstar)} = \frac{1}{\frac{c_s}{1-\alpha}-\psstar}
        \label{eq:psstar-foc}
    \end{equation}
    Clearly, the same argument suffices to imply $\psstar> \tilde{p}$. Furthermore, notice that $\frac{q'(p)}{q(p)}<0$, so we must have $\psstar>\frac{c_s}{1-\alpha} = \psind$ in order for the right side to be negative as well and for there to be an intersection. However, when $\alpha\geq 1-c_s\iff \psind \geq 1$, the independent seller cannot achieve positive payoff at any price, thus their optimal price would be $\psstar=1\leq \psind$ (because $q(1)=\pi_{s,s}(1)=0$ per our normalization choices).

    Finally, we turn our attention to showing $\pdagger>\tilde{p}$. The only nontrivial case is when $\pdagger<\min\{\psstar, \prstar\}$. Recall that $\pdagger$ was defined as the greatest price such that $\pi_{r,r}(\prstar) = \pi_{r,s}(\pdagger)$. Suppose for contadiction $\pdagger<\tilde{p}$; this implies $\pi_{r,s}(\tilde{p})>\pi_{r,r}(\prstar)$. However, since $\pi_{r,s}$ is continuous with $\pi_{r,s}(1)=0$, this implies there must exist some ${\pdagger}'\in(\tilde{p}, 1)$ such that $\pi_{r,s}({\pdagger}')=\pi_{r,r}(\prstar)$. However, this implies ${\pdagger}'>\tilde{p}>\pdagger$, contradicting the maximality of $\pdagger$ and concluding the proof.
\end{proof}
\setcounter{corollary}{0}
\begin{corollary}
    The payoff functions $\pi_{r,r}$, $\pi_{s,s}$, and $\pi_{r,s}$ are \emph{unimodal} on $[0,1]$, in the sense that they have only one critical point.
    \label{cor:unimodality}
\end{corollary}
\begin{proof}
    This follows from the first order conditions above; in each, the function on the right side is strictly increasing in $p$, and the function on the left is strictly decreasing in $p$ because
    \[\frac{\partial }{\partial p}\left(\frac{q'(p)}{q(p)}\right) = \frac{q(p)q''(p)-(q'(p))^2}{q^2(p)}<0\]
    Thus, the intersection price must be unique which is equivalent to the payoff function being unimodal.
\end{proof}
\subsection{Key prices and relative ordering}
\label{appendix:key-prices-relative-ordering}

In this section we derive the key prices and fees introduced in Sections \ref{sec:additional_notation} and \ref{sec:nash-equilibrium-existence}. We begin with the indifference prices --
\begin{align}
    \pi_{r,r}(\prind)
     & = \pi_{r,s}(\prind)\nonumber                      \\
    (\prind - c_r)q(\prind)
     & = \alpha\prind q(\prind)\nonumber                 \\
    \prind
     & = \frac{c_r}{1-\alpha}\label{eq:prind-derivation}
\end{align}
\begin{align}
    \pi_{s,s}(\psind)
     & = \pi_{s,r}(\psind)\nonumber                      \\
    ((1-\alpha)\psind - c_s)q(\psind)
     & = 0\nonumber                                      \\
    \psind
     & = \frac{c_s}{1-\alpha}\label{eq:psind-derivation}
\end{align}

Next, we will derive the fees that demarcate the changes in ordering between the prices. Note that there is no closed form for $\alphasdagger$. The simplest are $\alpharstar$ and $\alphasstar$ --
\begin{equation}
    \psind = \frac{c_s}{1-\alpha}
    \leq \prstar \iff  1-\alpha\geq \frac{c_s}{\prstar}\iff \alpha \leq 1-\frac{c_s}{\prstar} \equiv \alpharstar
    \label{eq:alpharstar-derivation}
\end{equation}
\begin{equation}
    \psind = \frac{c_s}{1-\alpha}
    \leq \psstar \iff  1-\alpha\geq \frac{c_s}{\psstar}\iff \alpha \leq 1-\frac{c_s}{\psstar} \equiv \alphasstar
    \label{eq:alphasstar-derivation}
\end{equation}
Next we find an expression for $\alphaopt$. Indeed, considering the functional form of the first order conditions \eqref{eq:prstar-foc} and \eqref{eq:psstar-foc}, we clearly have
\begin{equation}\prstar = \psstar \iff \frac{c_s}{1-\alphaopt} = c_r \iff \alphaopt = 1-\frac{c_s}{c_r}\label{eq:alphaopt-derivation}\end{equation}
To determine the ordering, note that as $\alpha\downarrow 0$, $\psind \downarrow c_s$, which implies by the first order conditions that $\psstar\leq \prstar$ when $\alpha=0$. Thus, the ordering must be $\psstar\leq \prstar\iff \alpha\leq \alphaopt$.

Finally, it remains to find a closed form for $\alphardagger$. Note that $\alphardagger$ determines the relative ordering of both $\prind$ and $\prstar$ as well as $\pdagger$ and $\prstar$. Thus in order to show $\alphardagger$ is well-defined, we begin by proving the following lemma --
\begin{lemma}
    $\prind\leq \prstar\iff \pdagger \leq \prstar$.
    \label{lemma:prind-prstar-pdagger-ordering}
\end{lemma}
\begin{proof}
    For the forward direction suppose $\prind \leq \prstar$ and suppose for contradiction $p^\dagger > \prstar \geq \prind$. This implies $\pi_{r,r}(\pdagger)> \pi_{r,s}(\pdagger) = \pi_{r,r}(\prstar)$. However, this is clearly a contradiction since $\prstar$ is the unique maximizer of $\pi_{r,r}$.

    For the reverse direction, suppose $\pdagger \leq \prstar$ and suppose for contradiction $\prind > \prstar$. Note $\pi_{r,s}(\pdagger) = \pi_{r,r}(\prstar)$, and since we have $\tilde{p}\leq \pdagger \leq \prstar$ by Lemma \ref{lemma:ptilde-ordering}, it follows that
    \[\pi_{r,s}(\prstar)\leq \pi_{r,s}(p^\dagger) = \pi_{r,r}(\prstar)\]
    However, this contradicts the definition of $\prind$ since $\prstar<\prind$ but $\pi_{r,s}(\prstar)\leq \pi_{r,r}(\prstar)$. This concludes the proof of the reverse direction.
\end{proof}

In light of the lemma, it's sufficient to define $\alphardagger$ as follows --
\begin{equation}\prind \leq \prstar \iff \frac{c_r}{1-\alpha} \leq \prstar \iff 1-\alpha \leq \frac{c_r}{\prstar}\iff \alpha \leq 1-\frac{c_r}{\prstar} \equiv \alphardagger
    \label{eq:alphardagger-derivation}\end{equation}

Since at $\alpha = \alphardagger$ we have $\pdagger$ is well-defined per Lemma \ref{lemma:prind-prstar-pdagger-ordering} (indeed, it is equal to $\prstar$), it immediately follows that $\pdagger$ is well-defined for all $\alpha>\alphardagger$ since $\pi_{r,s}$ is strictly increasing in $\alpha$.

\subsection{Relative ordering of fees}
\label{appendix:relative-ordering-fees}

For some of the fees, the relative ordering follows immediately from the closed form expressions, for instance
\begin{equation}\alphaopt<\alpharstar\iff 1-\frac{c_s}{c_r}<1-\frac{c_s}{\prstar}\iff c_r<\prstar
    \label{eq:alphaopt-alpharstar}
\end{equation}
\begin{equation}\alphardagger<\alpharstar\iff 1-\frac{c_r}{\prstar}<1-\frac{c_s}{\prstar}\iff c_s<c_r
    \label{eq:alphardagger-alpharstar}
\end{equation}

As motivated above, some of the fees' ordering depends on the seller optimal feasibility cost $\csstar$. Indeed, we have that
\begin{equation}\alphasstar\leq \alphardagger \iff 1-\frac{c_r}{\prstar}\leq 1-\frac{c_s}{c_r} \iff c_s \leq \frac{c_r^2}{\prstar}\equiv \csstar \label{eq:selleroptcost-derivation-i}\end{equation}
\begin{equation}\alphardagger\leq \alphaopt \iff  1-\frac{c_r}{\prstar}\leq 1-\frac{c_s}{c_r} \iff c_s \leq \frac{c_r^2}{\prstar}= \csstar\label{eq:selleroptcost-derivation-ii}\end{equation}

Finally, we turn our attention to $\alphasdagger$ which we define as the largest fee satisfying $\pi_{r,s}(\psind, \alphasdagger) = \pi_{r,r}(\prstar)$. Indeed, notice that we can write
\[\pi_{r,s}(\psind(\alpha)) = \alpha\psind q(\psind) = \left(1-\frac{c_s}{\psind}\right)\psind q(\psind) = (\psind - c_s)q(\psind) \]
where we now view $\psind(\alpha) = \frac{c_s}{1-\alpha}$ as a strictly increasing function of $\alpha$, that attains values
\[\psind(\alphaopt) = \frac{c_s}{1-(1-\frac{c_s}{c_r})} = c_r;\qquad \psind(1-c_s) = 1\]
By continuity every price in the interval $[c_r, 1]$ is attained by $\psind(\alpha)$ as $\alpha$ varies from $\alphaopt$ to $1-c_s$. With this in mind, notice that for all $p\in [c_r, 1]$ we have
\[\pi_{r,r}(p)= (p-c_r)q(p)<(p-c_s)q(p)= \pi_{r,s}(\psind(\alpha)=p)\]
In particular, this implies $\pi_{r,s}(\psind(\alpha))$ is a continuous function with values ranging from $0$ to $\pi_{r,s}(\psind(\alpha)=\prstar)>\pi_{r,r}(\prstar)$, which implies $\alphasdagger\in [0, 1-c_s]$ is well-defined.

We can further lower bound $\alphasdagger$ by recalling from Corollary \ref{cor:unimodality} that $\pi_{r,s}$ is unimodal. Thus, for any $\hat{\alpha}$ such that $\pi_{r,s}(\psind(\hat{\alpha}))>\pi_{r,r}(\prstar)$, we can conclude $\alphasdagger>\hat{\alpha}$. Since by construction we have $\psind(\alpharstar)=\prstar$, this therefore implies $\alphasdagger>\alpharstar$. Per \eqref{eq:alphaopt-alpharstar} and \eqref{eq:alphardagger-alpharstar}, this further implies $\alphasdagger>\alphaopt, \alphardagger$, and indeed if $c_s\leq \csstar$ we also have $\alphasdagger>\alphasstar$. By continuity of $\pi_{r,s}$, the following implication is immediate --

\begin{corollary}
    Suppose $\hat{\alpha}$ satisfies $\pi_{r,s}(\psind(\hat{\alpha}))>\pi_{r,r}(\prstar)$. Then, for $\alpha\in [\hat{\alpha}, 1]$, we have $\alpha\leq \alphasdagger\iff \psind\leq \pdagger$.
    \label{cor:alphasdagger-ordering}
\end{corollary}
When deriving Nash equilibria, we will find ourselves in scenarios like Corollary \ref{cor:alphasdagger-ordering} where we have an appropriate lower bound on $\alpha$, thus for all intents and purposes we will be able to think of the relation $\alpha\leq \alphasdagger \iff \psind\leq \pdagger$ as true.

\section{Detailed equilibria derivation}
\label{appendix:equilibria_derivation}

This section is dedicated to deriving the Nash equilibria of \eqref{eq:nash-equilibria-i} and \eqref{eq:nash-equilibria-ii}. For clarity of exposition, we outline conditions that any Nash equilibrium price configuration $(p_r, p_s)\in \Scal$ must satisfy. By definition of a Nash equilibrium, the following must hold:
\begin{enumerate}[(I)]
    \item $\sup_{p_r'<p_r}\pi_r(p_r', p_s) \leq \pi_r(p_r, p_s)$
    \item $\sup_{p_r'>p_r}\pi_r(p_r', p_s) \leq \pi_r(p_r, p_s)$
    \item $\sup_{p_s'<p_s}\pi_s(p_r, p_s') \leq \pi_s(p_r, p_s)$
    \item $\sup_{p_s'>p_s}\pi_s(p_r, p_s') \leq \pi_s(p_r, p_s)$
\end{enumerate}

\subsection{Shared-price equilibria $p_r = p = p_s$ in which the retailer fulfills demand}
\label{appendix:shared-price-retailer}

Condition IV always holds, as the independent seller achieves a payoff of zero whether they set some price $p' > p$ or they set price $p$ and allow the retailer to fulfill all demand.

Condition I requires the retailer would not prefer to set a price $p'<p$. Because $\pi_{r,r}$ is concave, it is increasing on $[0, \prstar]$, so condition I holds if and only if $p \leq \prstar$.

Condition II is equivalent to saying that the retailer would not prefer to increase their price, effectively allowing the independent seller to fulfill demand at price $p$ instead of the retailer fulfilling demand at price $p$, which by construction is equivalent to $p\geq \prind$.

Condition III posits that the independent seller would not prefer to undercut the retailer and fulfill demand themselves at some price $p'< p$. However, recall that in an equilibrium where the retailer fulfills demand, the independent seller achieves a payoff of $0$. Thus, we can rewrite condition III as:
\begin{align}
    \label{eq:a_sells_equilibria_condition_3}
    \sup_{p' \in [0, p)} \left[ \pi_{s, s} (p') \right] & \leq \pi_{s, r}(p) \\ \nonumber
    \sup_{p' \in [0, p)} \left[ \pi_{s, s} (p') \right] & \leq 0
\end{align}
Because $\pi_{s,s}$ is a concave unimodal function uniquely maximized at $\psstar$, we have that
\[\sup_{p'\in [0, p)}\left[ \pi_{s, s} (p') \right] =\begin{cases}
        \pi_{s, s} (p)       & p \leq \psstar \\
        \pi_{s, s} (\psstar) & p > \psstar
    \end{cases}\]
If $\psstar < p <  1$, then $0 < \pi_{s, s}(p) \leq \sup_{p' \in [0, p)} \left[ \pi_{s, s} (p') \right]$ and the condition never holds. On the other hand, if $p \leq \psstar$, then \eqref{eq:a_sells_equilibria_condition_3} reduces to
\begin{equation}
    \pi_{s, s} (p)                 \leq 0 =\pi_{s,r}(p)\iff p \leq \frac{c_s}{1 - \alpha} = \psind
\end{equation}
where we ignored the $1 \leq p$ because $p \in [0, 1]$.

Combining the conditions, we have shown that an equilibrium price in which the retailer fulfills demand must satisfy $p\in \left[\max\{0, \prind\}, \min\{1, \psind, \prstar, \psstar\}\right]$. However, because $c_s< c_r$, we have
\begin{equation}\psind = \frac{c_s}{1-\alpha} < \frac{c_r}{1-\alpha} = \prind\label{eq:psind-less-than-prind}\end{equation}
so the interval is always empty and there exist no equilibria in this case.

\subsection{Shared-price equilibria $p_r = p = p_s$ in which the independent seller fulfills demand}
\label{appendix:shared-price-is}

Condition II always holds, as the retailer does not fulfill any demand whether they set some price $p' > p$ or price $p$ and thus achieve the same payoff either way.

Condition III requires that the independent seller would not prefer to fulfill demand at $p'< p$. Because $\pi_{s,s}$ is a concave unimodal function maximized at $\psstar$, condition III is equivalent to $p \leq \psstar$.

Condition IV requires that the independent seller would not prefer that the retailer fulfill demand, in other words $p\geq \psind$.

Finally, condition I posits that the retailer would not prefer to undercut the independent seller and fulfill demand themselves at some price $p' < p$. Formally, we can rewrite condition I as
\[\sup_{p' \in [0, p)} \left[ \pi_{r, r} (p') \right] \leq \pi_{r, s}(p) \]
Because $\pi_{r,r}$ is a concave unimodal function, we have that
\begin{equation}
    \sup_{p' \in [0, p)}\left[ \pi_{r, r} (p') \right] = \begin{cases}
        \pi_{r,r}(p)       & p \leq \prstar \\
        \pi_{r,r}(\prstar) & p > \prstar
    \end{cases}\label{eq:maximizer-of-retailers-payoff}\end{equation}
\subsubsection{Equilibria where $p\leq \prstar$}
Consider the first case of \eqref{eq:maximizer-of-retailers-payoff} where $\sup_{p' \in [0, p)}\left[ \pi_{r, r} (p') \right] = \pi_{r,r}(p)$. Condition I then simply becomes $\pi_{r,r}(p)\leq \pi_{r,s}(p)$, which by construction is equivalent to $p\leq \prind$. Thus, we have found a continuum of equilibria $p\in \left[\max\{0, \psind\}, \min\{1, \prind,\prstar, \psstar\}\right]$, so long as the interval is nonempty. Note that $\psind \geq 0$ and $\prstar\leq 1$ always, so we can drop those.
\paragraph{Conditions under which the interval is nonempty.} We showed in \eqref{eq:psind-less-than-prind} that $\psind < \prind$ always, and $\psind \leq \prstar\iff \alpha\leq \alpharstar$. Thus the interval of prices is nonempty if and only if $\alpha \leq \alpharstar$.
\paragraph{Conditions under which $\prind$ is the upper bound.}
Recall that $\prind\leq \prstar\iff \alpha\leq \alphardagger$, and $\prind \leq \psstar\iff \alpha \leq \alphasstar$. Thus, we need $\alpha\leq \min\{\alphardagger, \alphasstar\}$ in order for $\prind$ to be the upper bound of the interval. Note also that $\alphasstar\leq \alphardagger\iff c_s \leq \csstar$.
\paragraph{Conditions under which $\psstar$ is the upper bound.}
We have that $\psstar \leq \prstar\iff \alpha\leq\alphaopt$ and $\psstar\leq \psind \iff\alpha\geq \alphasstar$, thus this case occurs if $\alpha\in[\alphasstar, \alphaopt]$. Also, recall that $\alphasstar\leq \alphaopt \iff c_s\leq \csstar$.

\paragraph{Conditions under which $\prstar$ is the upper bound.} Since we need both $\alpha>\alphaopt$ and $\alpha>\alphardagger$ in order for $\prstar$ to b the upper bound, this occurs whenever $\alpha\in \left[\max\left\{\alphaopt, \alphardagger \right\}, \alpharstar\right]$.

\paragraph{Conclusion and summary of equilibria.}
We have so far found the following equilibria:

\noindent\textbf{Case (i)}, $c_s \leq \csstar$:
\begin{equation}
    p\in
    \begin{cases}
        \left[\psind, \prind\right]  & \alpha
        \leq \alphasstar                                                     \\
        \left[\psind, \psstar\right] & \alphasstar\leq \alpha \leq \alphaopt \\
        \left[\psind, \prstar\right] & \alphaopt\leq \alpha \leq \alpharstar \\
    \end{cases}
    \label{eq:appendix-b2-case-i-part-1}\end{equation}

\noindent\textbf{Case (ii)}, $c_s \geq \csstar$:
\begin{equation}
    p\in
    \begin{cases}
        \left[\psind, \prind\right]  & \alpha
        \leq \alphardagger                                                       \\
        \left[\psind, \prstar\right] & \alphardagger\leq \alpha \leq \alpharstar \\
    \end{cases}
    \label{eq:appendix-b2-case-ii-part-1}\end{equation}
As we can see from \eqref{eq:appendix-b2-case-i-part-1} and \eqref{eq:appendix-b2-case-ii-part-1}, there only exists an equilibrium $(\psstar, s)$ when $c_s\leq \csstar$, which is indeed where the $\csstar$ notation comes from.
\subsubsection{Equilibria where $p\geq \prstar$}
\label{appendix:equilibria-p-geq-prstar}
Now, consider the second case of \eqref{eq:maximizer-of-retailers-payoff} where $\sup_{p' \in [0, p)}\left[ \pi_{r, r} (p') \right] = \pi_{r,r}(\prstar)$. The largest price for which $\pi_{r, s}(p)$ is equal to $\pi_{r, r}(\pastar)$ was defined in Section~\ref{sec:additional_notation} to be $p^\dagger$. Thus for an equilibrium in this case, we need that $p\in \left[\max\{0, \prstar, \psind\}, \min \{1, p^\dagger, \psstar\}\right]$. Note that $p^\dagger\leq 1$ and $\prstar, \psind \geq 0$ always, so we can drop those conditions.

Recall that $\psind \geq \prstar \iff \alpha \geq \alpharstar$. We won't explicitly simplify the condition $\psstar\leq p^\dagger$, but note in light of Lemma \ref{lemma:ptilde-ordering} we have
\[\psstar\leq \pdagger \iff \pi_{r,r}(\prstar)=\pi_{r,s}(\pdagger)\leq \pi_{r,s}(\psstar)\]

\paragraph{Interval is $[\psind, \min\left\{\psstar,p^\dagger\right\}]$.} Recall that $\psind\leq \psstar$ always, and $\alpha\in [\alpharstar, \alphasdagger] \implies \psind \leq p^\dagger$ per Corollary \ref{cor:alphasdagger-ordering}. Therefore, the interval is feasible as long as $\alpha\in\left[\alpharstar, \alphasdagger\right]$, which we showed above is always nonempty.

\paragraph{Interval is $[\prstar, \psstar]$.} Recall that $\prstar\leq \psstar \iff \alpha \geq \alphaopt$, so $[\prstar, \psstar]$ is nonempty whenever $\alpha \in \left[ \alphaopt, \alpharstar\right]$. Note we showed in \eqref{eq:alphaopt-alpharstar} that $\alphaopt\leq \alpharstar$ always.

\paragraph{Interval is $[\prstar, p^\dagger]$.} We have
$[\prstar, p^\dagger]$ is nonempty so long as $\alpha \in \left[\alphardagger, \alpharstar\right]$. Note we showed in \eqref{eq:alphardagger-alpharstar} that $\alphardagger\leq \alpharstar$ always.

\paragraph{Equilibria when $c_s \geq \csstar$.} We've so far shown that $p \in [\prstar, \psstar]$ or $p \in [\prstar, p^\dagger]$ could be an equilibrium if $\alpha \in \left[\min\left\{\alphaopt, \alphardagger\right\}, 1-c_s\right]$. In \eqref{eq:appendix-b2-case-i-part-1} and \eqref{eq:appendix-b2-case-ii-part-1}, we subdivided the equilibria into two cases, when $c_s \leq \csstar$ and $c_s \geq \csstar$. Though we have not explicitly simplified the condition $\psstar \leq p^\dagger$, it would be prudent to consider what equilibria exist when $\alpha$ is between $\alphardagger$ and $\alphaopt$.

Suppose $c_s \geq \csstar$ and $\alpha \in \left[\alphaopt, \alphardagger\right]$. Then, we claim this implies $p^\dagger < \psstar$. Suppose for contradiction $\psstar\leq \pdagger$. Then, since $\alpha<\alphardagger\leq\alphasstar$, we have $\prind\leq\psstar< \pdagger$, which implies
\[\pdagger>\prind \implies \pi_{r,r}(\pdagger)>\pi_{r,s}(\pdagger)=\pi_{r,r}(\prstar)\]
contradicting the optimality of $\prstar$. Now recall that earlier, we found that for all $\alpha\leq \alphardagger$, the equilibria are $p\in \left[\psind, \prind\right]$. However, we have that
\[\tilde{p}\leq p^\dagger\leq \prind \iff \pi_{r,s}(p^\dagger)\geq \pi_{r,s}(\prind)\iff \pi_{r,r}(\prstar)\geq \pi_{r,r}(\prind)\]
which always holds, thus we have found no new equilibria in this region of parameter space since $[\psind, \prind]\supseteq [\prstar, \pdagger]$.

\paragraph{Equilibria when $c_s \leq \csstar$.} Now suppose $c_s \leq \csstar$ and $\alpha\in \left[\alphardagger, \alphaopt\right]$, we claim this implies $\psstar\leq \pdagger$. Suppose for contradiction $\pdagger<\psstar$. Since $\alpha\geq\alphardagger$, we have $\prstar\leq\pdagger<\psstar$, but this contradicts the fact that $\alpha\leq \alphaopt \iff \psstar\leq \prstar$. Thus, in this region of parameter space, we have found equilibria $p\in [\prstar, \psstar]$, but note that $\alpha \leq \alphaopt \leq \alpharstar$ which implies $\psind\leq \prstar$. Therefore we have found no new equilibria once again in this region once again because before the equilibria we found were $p\in [\psind, \prstar]\supseteq [\psstar,\prstar]$.

\paragraph{Conclusion and summary of equilibria.} Combining with the equilibria from the first case from \eqref{eq:appendix-b2-case-i-part-1} and \eqref{eq:appendix-b2-case-ii-part-1}, we have the following shared-price equilibria where the independent seller fulfills --

\noindent\textbf{Case (i)}, $c_s \leq \csstar$:
\begin{equation}
    p\in
    \begin{cases}
        \left[\psind, \prind\right]                               & \alpha
        \leq \alphasstar                                                                                    \\
        \left[\psind, \psstar\right]                              & \alphasstar\leq \alpha \leq \alphaopt   \\
        \left[\psind, \min\left\{\psstar,p^\dagger\right\}\right] & \alphaopt\leq \alpha \leq \alphasdagger \\
    \end{cases}
    \label{eq:nash-equilibria-case-i}\end{equation}

\noindent\textbf{Case (ii)}, $c_s \geq \csstar$:
\begin{equation}
    p\in
    \begin{cases}
        \left[\psind, \prind\right]                               & \alpha
        \leq \alphardagger                                                                                      \\
        \left[\psind, \min\left\{\psstar,p^\dagger\right\}\right] & \alphardagger\leq \alpha \leq \alphasdagger \\
    \end{cases}
    \label{eq:nash-equilibria-case-ii}\end{equation}

Note that $\alpha\geq \alphardagger$ whenever $\pdagger$ is a possible equilibrium, which is sufficient for $\pdagger$ to be well-defined per Appendix \ref{appendix:key-prices-relative-ordering}.

\subsection{Equilibria where $p_r < p_s$}
\label{appendix:unequal-price-retailer}
Condition IV is always satisfied because for all $\varepsilon > 0$, $p_s + \varepsilon > p_s > p_r$, the independent seller achieves a payoff of $0$ both at price $p_s$ and price $p_s + \varepsilon$. Together, conditions I and II imply that the retailer must set price $\prstar$ so as to maximize their payoff function. Otherwise, if they are not at the optimal price, they would be able to increase their payoff by changing their price.

Furthermore, condition II requires that the retailer would not achieve a higher payoff by deviating to some price $p_r > p_s$ so as to let the independent seller fulfill demand in equilibrium. Formally, we need that
\[ \pi_{r,s}(p_s)\leq \pi_{r,r}(\prstar) \]
However, by construction this is equivalent to $p_s\geq p^\dagger$ (note that the smaller root is less than $\tilde{p}$ and therefore $\prstar$ by Lemma \ref{lemma:ptilde-ordering}). Note that we also must have that $p_s\geq \prstar$, and we have $\prstar \leq p^\dagger \iff \alpha \geq \alphardagger$.

We move on to condition III now. Recall that for any $p_s > p_r$ we have that $\pi_s(p_r, p_s) = 0$. Thus, condition III simplifies to
\[\max_{p_s'\in [0, \prstar]}\left[\pi_s(\prstar, p_s')\right]\leq 0\]
The maximum of the independent seller's payoff on the interval $[0, \prstar]$ is attained at
\[\arg\max_{p_s'\in [0, \prstar]}\left[\pi_s(\prstar, p_s')\right]  = \begin{cases}
        \prstar & \prstar \leq \psstar \\
        \psstar & \prstar > \psstar
    \end{cases}\]
because $\pi_s(\cdot, p_s')$ is maximized at $\psstar$. Consider the second case when $\prstar > \psstar>\psind$. Then, there always exists some price in $(\psind, \prstar)$ where the independent seller can make positive payoff so they would prefer to undercut the independent seller. Thus there are no equilibria where $\prstar > \psstar$.

On the other hand, $\arg\max_{p_s'\in [0, \prstar]}\left[\pi_s(\prstar, p_s')\right] = \prstar$ if $\prstar \leq \psstar$, which is equivalent to $\alpha \geq \alphaopt$. In order for this to be an equilibrium, we need further that $\psind \geq \prstar\iff \alpha\geq \alpharstar$ so that the independent seller does not prefer to undercut the retailer.

We showed in \eqref{eq:alphardagger-alpharstar} that $\alphardagger< \alpharstar$, so $\alpha \geq \alpharstar$ is sufficient to imply that $\prstar < p^\dagger$. Additionally, recall from \eqref{eq:alphaopt-alpharstar} that $\alphaopt\leq \alpharstar$. In summary, we have a continuum of equilibria $(\prstar, p_s \geq  p^\dagger)$ whenever $\alpha \geq\alpharstar$.

\subsection{Equilibria where $p_s < p_r$}
\label{appendix:unequal-price-is}

In this case, condition II is always satisfied; the retailer does not fulfill any demand either way as, for all $\varepsilon > 0$, $p_r + \varepsilon > p_r > p_s$. As in the previous case, conditions III and IV together imply that the independent seller must set price $\psstar$. Condition IV also requires that the independent seller does not prefer to let the retailer fulfill demand in equilibrium, in which case they would earn a payoff of $0$. To satisfy this, we need only that the independent seller sets a price $p_s \geq c_s$, which always holds for price $\psstar$.

Thus, all that remans to check is condition I. Simplifying the condition, we need that
\[\max_{p_r'\in [0,\psstar]}\left[\pi_{r,r}(p_r')\right]\leq \pi_{r,s}(\psstar)\]
Once again, there are two cases for the choice of $p_r'$ that maximizes the retailer's payoff function on the interval $[0, \psstar]$ --
\[\arg\max_{p_r'\in [0, \psstar]}\left[\pi_r(p_r', \psstar)\right]  = \begin{cases}
        \psstar & \psstar \leq \prstar \\
        \prstar & \psstar > \prstar
    \end{cases}\]
Consider the second case where $\arg\max_{p_r'\in [0, \psstar]}\left[\pi_r(p_r', \psstar)\right]  =\prstar $. Recall that $\psstar > \prstar \iff \alpha > \alphaopt$. Recalling Lemma \ref{lemma:ptilde-ordering}, condition I then simplifies to
\begin{align}
    \pi_{r,r}(\prstar)
     & \leq \pi_{r,s}(\psstar)\nonumber \\
    \pi_{r,s}(p^\dagger)
     & \leq \pi_{r,s}(\psstar)\nonumber \\
    \psstar
     & \leq p^\dagger\nonumber
\end{align}
Note that $\alpha >\alphasdagger\implies p^\dagger< \psind< \psstar$ by construction of $\alphasdagger$, thus $\alpha \leq \alphasdagger$ is necessary for condition I to hold. Thus, we have a continuum of equilibria $(p_r > \psstar, \psstar)$ whenever $\alpha \in \left[\alphaopt, \alphasdagger\right]$ and $\psstar\leq p^\dagger$ holds. Additionally, recall that if $c_s \geq \csstar$, we showed above that $\psstar\leq p^\dagger$ can only hold when $\alpha \geq\alphardagger$.

Now consider the first case where $\arg\max_{p_r'\in [0, \psstar]}\left[\pi_r(p_r', \psstar)\right]  =\psstar $. Recall that by construction, $\psstar \leq \prstar \iff \alpha \leq \alphaopt$. On the other hand, we have that condition I becomes
\[\pi_{r,r}(\psstar)\leq \pi_{r,s}(\psstar) \iff \psstar\leq \prind \iff \alpha\geq \alphasstar\]
Thus, we have a continuum of equilibria $(p_r > \psstar, \psstar)$ whenever $\alpha \in \left[\alphasstar, \alphaopt\right]$. Note that it's not always the case that $\alphasstar \leq \alphaopt$, but when this holds, such equilibria exist.

\section{Refining and interpreting equilibria - proofs}
\label{appendix:equilibria-refinements}
\subsection{Admissibility}
\label{appendix:admissibility}
For the independent seller, playing $p_s< \psind$ is weakly dominated by playing a price $p_s = 1$, because their payoffs are equal if the retailer sets a price $p_r\leq p_s$ but for $p_r\in (p_s, 1)$ we have
\[\pi_{s,s}(p_s) = \pi_s(p_r, p_s) \leq \pi_s(p_r, 1) = 0\]
because they achieve negative payoff when fulfilling demand below their breakeven price. Playing a price $p_s > \psstar$ is dominated by playing $\psstar$ because, for all $p_r > p_s$,
\[\pi_{s,s}(p_s) = \pi_{s}(p_r, p_s) \leq \pi_{s}(p_r, \psstar) = \pi_{s,s}(\psstar) \]
where the inequality holds because $p_s > \psstar$. For any price $p_r \leq p_s$, we have
\[0 = \pi_{s,r}(p_s) = \pi_s(p_r, p_s)\leq \pi_s(p_r, \psstar)\]
which implies weak dominance as desired. However, consider for instance two prices $\psind<p_{s_1}< p_{s_2}<\psstar$, we claim neither dominates the other. If $p_r>p_{s_2}$,
\[\pi_{s,s}(p_{s_1}) = \pi_s(p_r, p_{s_1}) < \pi_s(p_r, p_{s_2}) = \pi_{s,s}(p_{s_2})\]
because $p_{s_1}< p_{s_2}<\psstar$. On the other hand if $p_r \in (p_{s_1}, p_{s_2})$,
\[\pi_{s,s}(p_{s_1}) = \pi_s(p_r, p_{s_1}) > \pi_s(p_r, p_{s_2}) = 0\]
Thus, in any region of parameter space, the only admissible strategies for the independent seller are $p_s\in (\psind, \psstar)$, except for a technical edge case when $\alpha\geq 1-c_s \implies \psind\geq 1$ which we will address later.

A similar result holds for the retailer, except it's not always the case that $\prind\leq\prstar$ (which occurs precisely when $\alpha \leq \alphardagger$ by construction). We claim that $p_r = \min\left\{\prind, \prstar\right\}$ dominates all $p_r' < p_r$. Indeed, for $p_s \leq p_r'$, the retailer's payoff is the same. For $p_s \geq p_r$, the retailer's payoff is
\[\pi_{r,r}(p_r') = \pi_r(p_r', p_s) \leq \pi_r(p_r, p_s) = \pi_{r,r}(p_r)\]
where the inequality holds because $p_r'< p_r\leq \prstar$. For $p_s \in (p_r', p_r)$, we have
\[\pi_{r,r}(p_r') = \pi_r(p_r', p_s) \leq \pi_r(p_r, p_s) = \pi_{r,s}(p_s)\]
where the inequality holds because $p_r' < p_s < p_r\leq \prind$.

Similarly, $p_r = \max\left\{\prind, \prstar\right\}$ dominates all $p_r' > p_r$. Indeed, for $p_s\leq  p_r$, the retailer's payoff is the same. For $p_s \geq p_r'$, the retailer's payoff is
\[\pi_{r,r}(p_r') = \pi_r(p_r', p_s) \leq \pi_r(p_r, p_s) = \pi_{r,r}(p_r)\]
where the inequality holds because $p_r'> p_r\geq \prstar$. For $p_s \in (p_r, p_r')$, we have
\[\pi_{r,s}(p_s) = \pi_r(p_r', p_s) \leq \pi_r(p_r, p_s) = \pi_{r,r}(p_r)\]
where the inequality holds because $p_r'> p_s >  p_r\geq \prind$. However, we can show fairly easily that for any $\min\left\{\prind, \prstar\right\}<p_{r_1}< p_{r_2}<\max\left\{\prind, \prstar\right\}$, neither weakly dominates the other. Suppose $\prind<\prstar$, then for $p_s > p_{r_2}$ we have
\[\pi_r(p_{r_1}, p_s) = \pi_{r,r}(p_{r_1}) \leq \pi_r(p_{r_2}, p_s) = \pi_{r,r}(p_{r_2})\]
because $p_{r_1}< p_{r_2}<\prstar$. While, for, $p_s\in (p_{r_1}, p_{r_2})$, we have
\[\pi_{r,s}(p_{r_2}) = \pi_r(p_{r_2}, p_s) \leq\pi_r(p_{r_1}, p_s) =  \pi_{r,r}(p_{r_1}) \]
because $\prind< p_{r_1}< p_{r_2}$. Reversing the directions of the inequalities proves that there are no dominant strategies for any $\prstar<p_{r_1}<p_{r_2}<\prind$.

With this in mind, we will step through each of the Nash equilibria that we found in \eqref{eq:nash-equilibria-case-i} and \eqref{eq:nash-equilibria-case-ii} and remove those strategies that are weakly dominated.

First consider an equilibrium in the regime where $\alpha\leq \min\left\{\alphasstar, \alphardagger\right\}$. Here, the Nash equilibrium is $p\in [\psind, \prind]$. Because the retailer's admissible strategies fall in the interval $p_r\in [\prind, \prstar]$, the only admissible equilibrium price is $p = \prind$.

Now, suppose that $\alpha \in \left[\alphasstar, \alphaopt\right]$, in which case the Nash equilibria are $(p_r \geq \psstar, p_s = \psstar)\cup (p\in [\psind, \psstar])$. We showed that in this region of parameter space, we have $\psind\leq \psstar\leq \prstar, \prind$. Notably, this implies that the entire continuum of shared-price equilibria $p\in [\psind, \psstar]$ are inadmissible because it is entirely disjoint from the admissible region for the retailer, $p_r \in \left[\min\left\{\prind, \prstar\right\}, \max\left\{\prind, \prstar\right\}\right]$. Thus, the admissible equilibria are $(p_r \in \left[\prind, \prstar\right], p_s = \psstar)$ if $\alpha \in \left[\alphasstar, \alphardagger\right]$ and $(p_r \in \left[\prstar, \prind\right], p_s = \psstar)$ if $\alpha \in \left[\alphardagger, \alphaopt\right]$.

Now, suppose $\psstar\leq p^\dagger$ and $\alpha \in \left[\max\left\{\alphardagger, \alphaopt\right\},\alpharstar\right]$. As before, the Nash equlibria are $(p_r\geq \psstar, p_s = \psstar)\cup (p \in\left[\psind,\psstar\right])$. However, this time, we have $\psind\leq\prstar\leq \psstar\leq p^\dagger\leq \prind$. Here, the retailer's admissible prices are $p_r\in [\prstar, \prind]$, so the remaining admissible equilibria are $(p_r\in [\psstar, \prind], p_s = \psstar)\cup (p \in\left[\prstar,\psstar\right])$.

The arguments of the previous paragraph did not depend on the relative ordering of $p^\dagger$ and $\psstar$, and we know that $\prind \geq p^\dagger\geq \prstar$ because $\alpha \geq \alphardagger$. Thus the admissible equilibria when $\alpha \in \left[\max\left\{\alphardagger, \alphaopt\right\}, \alpharstar\right]$ and $p^\dagger \leq \psstar$ are simply the shared prices $p \in\left[\prstar,p^\dagger\right]$.

Consider the case when $\alpha \in \left[\alpharstar, \alphasdagger\right]$, where the equilibria are $(p_r = \prstar, p_s\geq p^\dagger) \cup (p \in [\psind, \min\{p^\dagger, \psstar\}])$. In this case, $\prstar\leq \psind \leq \psstar, p^\dagger\leq \prind$ always. If $\psstar\leq p^\dagger$ also, then the only admissible equilibria are $p\in [\psind, \psstar]$ as prices $p_s \geq \psstar$ are inadmissible. On the other hand, if $p^\dagger\leq \psstar$, the admissible equilibria are $(p_r = \prstar, p_s\in [p^\dagger, \psstar])\cup (p\in [\psind, p^\dagger])$.

Now suppose that $\alpha \in [\alphasdagger, 1-c_s]$, here the only Nash equilibrium is $(p_r = \prstar, p_s\geq p^\dagger)$. In this region of parameter space, we always have that $p^\dagger\leq \psind\leq \psstar\leq 1$ by construction of $\alphasdagger$, so the admissible Nash equilibria are $(p_r = \prstar, p_s\in [\psind, \psstar])$.

However, when $\alpha \geq 1-c_s \geq \alphasdagger$, we have that $\psind \geq 1$ and the Nash equilibrium is again given by $(p_r = \prstar, p_s\geq p^\dagger)$. In this limit, setting $p_s = 1$ weakly dominates setting any $p_s' < 1$ because $\psind \geq 1 > p_s'$ which implies that the independent seller would achieve negative payoff for each unit sold at price $p_s'$ and would prefer to achieve 0 payoff by setting $p_s = 1$. Thus, the only admissible equilibrium when $\alpha \geq 1-c_s$ is $(p_r = \prstar, p_s=1)$. Note that this edge case is technical and unimportant, as no matter what price the independent seller sets, the result remains the same -- when $\alpha$ is high, there is an equilibrium where the retailer fulfills demand.

Summarizing, we have narrowed our focus to the following admissible equilibria:

\noindent\textbf{Shared equilibria}, irrespective of ordering of $c_s$ and $\csstar$:
\begin{equation}
    (p_r, p_s) \in \begin{cases}
        (p_s, [\psind, p^\dagger]) \cup (\prstar,[p^\dagger, \psstar])\footnotemark & \left[\alpharstar\leq \alpha\leq \alphasdagger\right] \land [p^\dagger\leq \psstar] \\
        (p_s, [\psind, \psstar])                                                    & \left[\alpharstar\leq \alpha\leq \alphasdagger\right]\land [\psstar\leq p^\dagger]  \\
        (\prstar, [\psind, \psstar])                                                & \alphasdagger\leq \alpha \leq 1-c_s                                                 \\
        (\prstar, 1)                                                                & 1-c_s\leq \alpha                                                                    \\
    \end{cases}
    \label{eq:admissible-equilibria-shared}
\end{equation}
\footnotetext{This is the only admissible equilibrium that is relatively Pareto suboptimal. See Appendix \ref{appendix:pareto_optimality_proof} for further explanation.\label{footnote:pareto-suboptimal}}

\noindent\textbf{Case (i)}, $c_s \leq \csstar$:
\begin{equation}
    (p_r, p_s)\in \begin{cases}
        (\prind, \prind)                                                     & \alpha
        \leq 1 - 2c_r +c_s                                                                                                                                                 \\
        ([\prind, \prstar], \psstar)                                         & \alphasstar\leq \alpha \leq \alphardagger                                                   \\
        ([\prstar, \prind], \psstar)                                         & \alphardagger\leq \alpha \leq \alphaopt                                                     \\
        (p_s,\left[\prstar,\psstar\right])\cup ( [\psstar, \prind], \psstar) & \left[\alphaopt\leq \alpha \leq \alpharstar\right]\land \left[\psstar \leq p^\dagger\right] \\
        (p_s, \left[\prstar,p^\dagger\right])                                & \left[\alphaopt\leq \alpha \leq \alpharstar\right]\land \left[p^\dagger\leq \psstar\right]  \\
    \end{cases}
    \label{eq:admissible-equilibria-i}
\end{equation}

\noindent\textbf{Case (ii)}, $c_s \geq \csstar$:
\begin{equation}
    (p_r, p_s)\in\begin{cases}
        (\prind, \prind)                                                     & \alpha
        \leq \alphardagger                                                                                                                                                     \\
        (p_s,\left[\prstar,\psstar\right])\cup ( [\psstar, \prind], \psstar) & \left[\alphardagger\leq \alpha \leq \alpharstar\right]\land \left[\psstar \leq p^\dagger\right] \\
        (p_s, \left[\prstar,p^\dagger\right])                                & \left[\alphardagger\leq \alpha \leq\alpharstar\right]\land \left[p^\dagger\leq \psstar\right]   \\
    \end{cases}
    \label{eq:admissible-equilibria-ii}
\end{equation}

\subsection{Relative Pareto optimality}
\label{appendix:pareto_optimality_proof}

Consider the equilibria in \eqref{eq:admissible-equilibria-shared}, \eqref{eq:admissible-equilibria-i}, and \eqref{eq:admissible-equilibria-ii}, we will remove any that are Pareto suboptimal relative to other admissible Nash equilibria. There is one Pareto-suboptimal admissible equilibrium (indicated in the first case of \eqref{eq:admissible-equilibria-shared} by Footnote \ref{footnote:pareto-suboptimal}), when $\alpharstar\leq \alpha \leq \alphasdagger$ and $p^\dagger\leq \psstar$, in which case the retailer fulfilling demand at $\prstar$ yields them the same payoff as the independent seller fulfilling demand at $p^\dagger$. However, the latter earns the independent seller positive payoff while the former earns the independent seller a payoff of $0$. Thus, the equilibrium $(p_r = \prstar, p_s\in [p^\dagger, \psstar])$ is Pareto suboptimal relative to the shared-price equilibrium $p = p^\dagger$.

When $\alpha \leq \max\left\{\alphardagger, \alphaopt\right\}$ or $\alpha \geq 1-c_s$, there is only one admissible equilibrium. When $\alphasdagger\leq \alpha \leq 1-c_s$, despite there being a continuum of equilibria in some of these cases, they all result in the same outcome with the retailer fulfilling demand at $\prstar$. Similarly, when $\alphasstar\leq \alpha\leq \alphaopt$, all of the equilibria result in the independent seller fulfilling demand at $\psstar$. In the remaining cases, all equilibria are Pareto optimal, because the retailer prefers to push the price lower towards $\prstar$ and/or $\tilde{p}$ (their optimal price when the independent seller fulfills demand) and the independent seller prefers to push the price higher towards $\psstar$.

\section{Optimizing the referral fee - proofs}
\subsection{Nash equilibrium refinements}
\label{appendix:sequential_game_refinements}
First, we justify the claims made in Section \ref{sec:alpha-game-setup} about relative Pareto optimality and admissibility of equilibrium strategies for the sequential game.

Suppose there are two Nash equilibria of the sequential game $(\alpha_1, p_r, p_s)$ and $(\alpha_2, p_r', p_s')$ with $\alpha_1 < \alpha_2$. Because these are both equilibria, they must have the same payoff to the retailer. However, because the independent seller's payoff is monotonically decreasing in $\alpha$, $(\alpha_2, p_r', p_s')$ is less Pareto optimal than $(\alpha_1, p_r, p_s)$. Thus, Pareto optimality yields us the natural refinement that the retailer will choose the minimum $\alpha$ necessary to attain their maximum possible equilibrium payoff.

Let $\alpha_1 < \alpha_2$, we will produce a strategy for the independent seller in which the retailer achieves a higher payoff with $\alpha_1$ and another strategy in which the retailer achieves a higher payoff with $\alpha_2$. Consider the strategy: ``choose price $p_{s_1}$ if $\alpha \leq \frac{\alpha_1 + \alpha_2}{2}$, otherwise choose price $p_{s_2}$'' where $p_{s_1}, p_{s_2}$ are chosen such that $\pi_{r,s}(p_{s_1}, \alpha_1)> \pi_{r,s}(p_{s_2}, \alpha_2)$. Switching the strategy to play $p_{s_2}$ when $\alpha \leq \frac{\alpha_1 + \alpha_2}{2}$ and $p_{s_1}$ otherwise concludes the proof that both $\alpha_1, \alpha_2$ are admissible.

\subsection{Investigating equilibrium strategy profiles}
\label{appendix:equilibrium-strategy}
Despite the preponderance of possible equilibrium strategy profiles, examining \eqref{eq:body-effective-equilibria-i} and \eqref{eq:body-effective-equilibria-ii} shows there are only a few possible admissible, Pareto optimal outcomes in terms of who fulfills demand and what price they set. The only region of parameter space in which $\rho$ nontrivially affects the equilibrium outcomes is when $\alpha \in \left[\max\left\{\alphardagger, \alphaopt\right\}, \alphasdagger\right]$, because in this case there are a continuum of prices $p \in \left[\max\left\{\prstar, \psind\right\},\min\left\{\psstar, p^\dagger\right\}\right]$ at which the independent seller could fulfill demand in equilibrium. Intuitively, this means that though $\rho$ could be very complicated in general, this is the only region in which this generality actually has an effect on the payoffs and the real-world outcomes.

We'll denote as $\pireq(\rho(\alpha), c_r, c_s)$ the equilibrium payoff function for the retailer with strategy $\rho$ and $\piseq$ the analogous function for the independent seller. With this notation, we then have $\alphastar \equiv \min_\alpha[\arg\max_\alpha \pireq(\rho(\alpha), c_r, c_s)]$, and our eventual goal is to compute this $\alphastar$. Even with $c_r, c_s, \rho$ held constant these are still complicated piecewise functions of $\alpha$, so as a first step we will prove continuity of the equilibrium payoffs in $\alpha$.

\begin{proposition}
    For all $\rho$ that are continuous on $\alpha\in \left(\max\left\{\alphardagger, \alphaopt\right\}, \alphasdagger\right)$,\footnote{Throughout this paper we consider only pure strategies. However, consider for a moment mixed strategies $\mu_\alpha(p)$ such that for each $\alpha$, $\mu_\alpha$ is a valid density over equilibrium prices. Then, the analysis easily generalizes if we define $\mathbb{E_{\mu_\alpha}}[p]\equiv \rho(\alpha)$.} $\pireq$ and $\piseq$ are continuous in $\alpha$.
    \label{thm:continuity-equilibrium-payoff}
\end{proposition}
\begin{proof}
    Examine the equilibrium prices in \eqref{eq:body-effective-equilibria-i} and \eqref{eq:body-effective-equilibria-ii}. Within each of the pieces of the function, $\pireq, \piseq$ are constant as there is only one price at which demand will be fulfilled in equilibrium, except for when $\alpha \in\left(\max\left\{\alphardagger, \alphaopt\right\}, \alphasdagger\right)$. In this interval, we have
    \[\pireq(\alpha, c_r, c_s, \rho) = \pi_{r,s}(p_s = \rho(\alpha); c_r, c_s),\qquad\piseq(\alpha, c_r, c_s, \rho) = \pi_{s,s}(p_s = \rho(\alpha); c_r, c_s)\]
    because $\rho$ is continuous in $\alpha$ on this interval by assumption and $\pi_{r,s}, \pi_{s,s}$ are continuous functions of the price $p_s$, this is a composition of continuous functions which is continuous.

    Thus, all that remains to check is that at each of the transitions between pieces, the lower limit is the same as the upper limit of $\pireq$ and $\piseq$. We know by construction that $\prind = \psstar$ when $\alpha = \alphasstar$. However, the remaining transitions are not quite as clear, as they depend on the value of $\rho$ when $\alpha = \max\left\{\alphardagger, \alphaopt\right\}$ and $\alpha = \alphasdagger$. We will handle each of these cases in turn.

    Now suppose that $c_s \geq \csstar$. When $\alpha = \alphardagger$,$\prstar = p^\dagger$ by construction (and $\prstar < p^\dagger$ for $\alpha > \alphardagger$). Because $\rho$ yields an admissible, Pareto optimal equilibrium price configuration for each $\alpha$ and $\pi_{r,s}$ is continous, the squeeze theorem tells us that
    \begin{equation}\pi_{r,s}(\prstar) = \lim_{\alpha \downarrow \alphardagger} \pi_{r,s}(\prstar) \leq\lim_{\alpha \downarrow \alphardagger} \pireq(\rho(\alpha)) \leq \lim_{\alpha \downarrow \alphardagger} \pi_{r,s}(p^\dagger) = \pi_{r,s}(\prstar)\label{eq:squeeze-theorem-application}\end{equation}
    Furthermore, recall that at $\alpha = \alphardagger$, $\prind = \prstar$, implying continuity because
    \begin{equation}
        \lim_{\alpha \uparrow\alphardagger} \pireq(\rho(\alpha)) = \lim_{\alpha \uparrow \alphardagger} \pi_{r,s}(\prind) = \pi_{r,s}(\prind) = \pi_{r,s}(\prstar) = \lim_{\alpha \downarrow \alphardagger} \pireq(\rho(\alpha))
    \end{equation}
    A similar application of the squeeze theorem shows continuity when $c_s \leq \csstar$ as $\alpha \to \alphaopt$, because by construction $\psstar \leq \prstar \iff \alpha \leq \alphaopt$.

    By construction, $\alphasdagger$ is the referral fee at which $\psind = p^\dagger$. Because $p^\dagger$ was defined satisfy $\pi_{r,s}(p^\dagger) = \pi_{r,r}(\prstar)$, and above $\alphasdagger$ the equlibrium outcome is the retailer fulfilling demand at $\prstar$, we can apply the squeeze theorem a final time similarly to show continuity as $\alpha \to \alphasdagger$.
\end{proof}
\begin{corollary}
    For $\alpha\in \left[\max\left\{\alphardagger, \alphaopt\right\}, \alphasdagger\right]$, let $\underline{\rho}(\alpha) \equiv \max\left\{\prstar, \psind\right\}$ and $\overline{\rho}(\alpha) \equiv \min\left\{\psstar, p^\dagger\right\}$. Then, $\pireq(\overline{\rho}(\alpha))\leq \pireq(\rho(\alpha)) \leq \pireq(\underline{\rho}(\alpha))$, and $\piseq(\underline{\rho}(\alpha))\leq \piseq(\rho(\alpha)) \leq \piseq(\overline{\rho}(\alpha))$.
    \label{thm:bounded-equilibrium-payoffs}
\end{corollary}
\begin{proof}
    Because $\rho$ is an equilibrium strategy profile, we know that for all $\alpha\in\left[\max\left\{\alphardagger, \alphaopt\right\}, \alphasdagger\right]$, $\rho(\alpha)\in \left[\max\left\{\prstar, \psind\right\},\min\left\{\psstar, p^\dagger\right\}\right]\subset \left[\tilde{p}, \psstar\right]$ per Lemma \ref{lemma:ptilde-ordering}. However, note that $\tilde{p}$ is the maximizer of the concave function $\pi_{r,s}$ and $\psstar$ is the maximizer of the concave function $\pi_{s,s}$. Thus, for any given $\alpha \in\left[\max\left\{\alphardagger, \alphaopt\right\}, \alphasdagger\right]$, the retailer achieves their maximum (minimum) payoff by choosing the lowest (highest) possible equilibrium price and the independent seller achieves their maximum (minimum) payoff by choosing the highest (lowest) possible equilibrium price.
\end{proof}

\subsection{Lower bounding the retailer's payoff}
\label{appendix:lower-bounding-retailer-payoff}
Following the argument from Corollary \ref{thm:bounded-equilibrium-payoffs}, we know given $\alpha\in \left[\max\left\{\alphardagger, \alphaopt\right\}, \alphasdagger\right]$, the retailer's payoff is minimized at the maximum possible price. In other words, the worst-case $\rho$ would be one that always chooses the maximum possible price when there are multiple equilibrium prices, precisely $\overline{\rho}(\alpha) \equiv \min\left\{\psstar, p^\dagger\right\}$ for $\alpha\in \left[\max\left\{\alphardagger, \alphaopt\right\}, \alphasdagger\right]$. Formally, for all $\rho$, we have
\[\max_\alpha \pireq(\alpha, c_r, c_s, \rho)\geq \max_\alpha \pireq(\alpha, c_r, c_s, \overline{\rho})\]
Thus, we turn our attention to maximizing the retailer's equilibrium payoff in the worst case with $\rho = \overline{\rho}$. First, we notice that since for all $\alpha$ we have
\begin{equation}\pi_{r,s}(\prind)=\pi_{r,r}(\prind) \leq \pi_{r,r}(\prstar) \label{eq:payoff-prind-prstar}\end{equation}
the maximizer of $\pireq$ is never in a piece where the independent seller fulfills demand at $\prind$. This follows because the retailer's equilibrium payoff is never greater than their payoff when setting the fee above $\alphasdagger$ and fulfilling demand themselves.


Because we are working with $\overline{\rho}$, the only remaining equilibrium outcomes that could maximize the retailer's payoff are ones where the independent seller fulfills demand at $\psstar$ or $p^\dagger$ and one where the retailer fulfills demand at $\prstar$. However, by construction $\pi_{r,s}(p^\dagger) = \pi_{r,r}(\prstar)$. Thus, all that remains is to understand if or when the retailer can achieve a greater payoff than $\pi_{r,r}(\prstar)$ in an equilibrium where the independent seller fulfills demand at $\psstar$. To this end, we introduce $\overline{\alpha}$ as the relatively Pareto-optimal maximizer of $\pi_{r,s}(\psstar)$, defined formally as follows --
\[\overline{\alpha}\equiv \min_\alpha\left[ \arg\max_\alpha\pi_{r,s}(\psstar,\alpha)\right]\]
First, we would like to understand if $\overline{\alpha}$ is in a region of paramter space where $\overline{\rho}(\overline{\alpha})=\psstar$. Indeed, we have the following --
\begin{lemma}
    If $\pi_{r,s}(\psstar, \overline{\alpha}) \geq \pi_{r,s}(p^\dagger)$, then $\overline{\alpha}\in [\alphasstar, \alphasdagger]$ and $\pireq(\overline{\alpha}, c_r, c_s, \overline{\rho})=\pi_{r,s}(\psstar,\overline{\alpha})$.
    \label{thm:alpha-in-support}
\end{lemma}
\begin{proof}
    To prove the upper bound, suppose for contradiction that $\overline{\alpha} > \alphasdagger$. Then, since $\psstar, \pdagger>\tilde{p}$ per Lemma \ref{lemma:ptilde-ordering}, we have
    \begin{equation}\pi_{r,s}(\psstar, \overline{\alpha})\geq \pi_{r,s}(\pdagger)\iff \pdagger\geq\psstar \label{eq:psstar-pdagger-condition}\end{equation}
    If $\overline{\alpha}\geq1-c_s$, then $\psstar=1$ so this is an immediate contradiction. Otherwise for $\overline{\alpha}\in\left(\alphasdagger,1-c_s\right)$, we have $\psind>\pdagger\geq\psstar$ which is a contradiction per Lemma \ref{lemma:ptilde-ordering}.

    For the lower bound, suppose for contradiction $\overline{\alpha}<\alphasstar$. Then, $\pi_{r,s}(\psstar, \overline{\alpha})\geq\pi_{r,r}(\prstar)\geq\pi_{r,r}(\prind)=\pi_{r,s}(\prind)$. By the optimality of $\psstar$ we have $\pi_{s,s}(\psstar, \overline{\alpha})> \pi_{s,s}(\prind)$ where the inequality holds strictly because $\overline{\alpha}<\alphasstar$ strictly. It follows from this that at $\overline{\alpha}$, $(\prind, s)$ cannot be a Nash equilibrium of the staying subgame, which is a contradiction per our earlier derivation. This concludes the proof that $\overline{\alpha}\in [\alphasstar, \alphasdagger]$.

    The second part of the lemma follows immediately because either $\overline{\rho}(\overline{\alpha})=\psstar$ when $\overline{\alpha}\in [\alphasstar,\alphaopt]$ or $\overline{\rho}(\overline{\alpha})=\min\left\{\psstar, \pdagger\right\} =\psstar$ per \eqref{eq:psstar-pdagger-condition}.
\end{proof}

Because of its dependence on $\psstar$, finding a simple closed form for $\overline{\alpha}$ is not possible in general, though we can still proceed with our goal of understanding when $\pi_{r,s}(\psstar, \overline{\alpha})\geq \pi_{r,r}(\prstar)$ by noting
\begin{equation}\pi_{r,s}(\psstar, \overline{\alpha})\geq \pi_{r,r}(\prstar)\iff \exists \alpha \text{ st } \pi_{r,s}(\psstar, \alpha)\geq \pi_{r,r}(\prstar)\label{eq:overline-alpha-equivalent-condition}\end{equation}
which is far simpler. Using this equivalent characterization, let's pause and prove the following sufficient condition:
\begin{lemma}
    If $c_s \leq \csstar$, $\pi_{r,s}(\psstar, \overline{\alpha}) \geq \pi_{r,s}(p^\dagger)$ always.
    \label{lemma:csstar-implies}
\end{lemma}
\begin{proof}
    Consider $\alpha = \alphardagger$, at which $\prind=\prstar=\pdagger$. Since $c_s\leq \csstar\implies \alphardagger\leq \alphaopt$, we have $\prstar \geq \psstar>\tilde{p}$. Thus,
    \[\pi_{r,s}(\psstar,\alphardagger)\geq \pi_{r,s}(\prstar,\alphardagger)=\pi_{r,s}(\pdagger, \alphardagger)=\pi_{r,r}(\prstar)\]
    which is sufficient to conclude the proof per \eqref{eq:overline-alpha-equivalent-condition}.
\end{proof}

The proof of Lemma \ref{lemma:csstar-implies} breaks down when $c_s \geq \csstar$ because we have now that $\alphaopt\leq \alphardagger$. When $c_s \geq \csstar$, there sometimes exists an $\alpha$ satisfying \eqref{eq:overline-alpha-equivalent-condition} and sometimes does not. Finally, we note that the equilibrium fee is particularly simple in the following case --

\begin{lemma}
    Suppose $\pi_{r,s}(\psstar, \alpha)\leq \pi_{r,r}(\prstar)$. Then, $\alphastar(c_r, c_s, \overline{\rho})=\alphardagger$.
\end{lemma}
\begin{proof}
    Notice that \eqref{eq:payoff-prind-prstar} holds with equality if and only if $\prind=\prstar \iff \alpha=\alphardagger$. Thus, $\alphardagger$ is the smallest $\alpha$ such that $\pireq(\alpha, c_r, c_s, \overline{\rho})=\pi_{r,r}(\prstar)$. Furthermore, for all $\alpha > \alphardagger$, we have $\pi_{r,s}(\psstar, \alpha)\leq \pi_{r,r}(\prstar)$ by assumption\footnote{Combining this with \eqref{eq:psstar-pdagger-condition} also implies that $\forall \alpha\geq \alphardagger$, $\pireq(\alpha, c_r, c_s, \overline{\rho})=\pi_{r,r}(\prstar)$}. Thus, $\alphardagger$ is indeed the minimum fee such that the retailer attains their maximum payoff, concluding the proof.
\end{proof}

Combining the lemmas, we can conclude that
\begin{equation}
    \max_\alpha \pireq(\alpha, c_r, c_s, \rho)\geq \max_\alpha \pireq(\alpha, c_r, c_s, \overline{\rho}) = \begin{cases}
        \pi_{r,s}(\psstar, \overline{\alpha}) & \exists \alpha \text{ s.t. } \pi_{r,s}(\psstar, \alpha)> \pi_{r,r}(\prstar) \\
        \pi_{r,r}(\prstar)                    & \text{else}
    \end{cases}
    \label{eq:lower-bound-retailer-payoff}
\end{equation}
We have shown that $\alphastar(\overline{\rho}) = \overline{\alpha}$ in the first case and $\alphastar(\overline{\rho}) = \max\left\{\alphardagger, \alphaopt\right\}$ otherwise. Ideally, we would be able to directly translate this to a tight bound on $\alphastar\geq \overline{\alpha}$ for all $\rho$ in the first case, however it is still possible that there exists some $\rho$ and some $\alpha'<\overline{\alpha}$ such that
\begin{equation}
    \max_\alpha\pi_{r,s}(\rho(\alpha)) = \pi_{r,s}(\rho(\alpha')) > \pi_{r,s}(\rho(\overline{\alpha}))
    \label{eq:counterxample-bounding-alpha-star}
\end{equation}
For instance, \eqref{eq:counterxample-bounding-alpha-star} could hold if $\rho(\alpha') = \prstar$ for some well-chosen $\alpha'<\overline{\alpha}$, but for all other $\alpha$, we set $\rho = \overline{\rho}$. If we impose some assumptions on the functional form of $\rho$, we may be able to translate \eqref{eq:lower-bound-retailer-payoff} into a tighter bound on $\alphastar$. However, we do not investigate further tightening the bound in this work as we have at least shown that the lower bound $\alphastar\geq \max\left\{\alphardagger, \alphaopt\right\}$ is tight whenever we are in the second case of \eqref{eq:lower-bound-retailer-payoff}.

With this, the maximum payoff that the independent seller could achieve in equilibrium of the sequential game is either by fulfilling demand at $p^\dagger$ or $\psstar$. If it is at $\psstar$, clearly they would prefer the lowest possible $\alpha$. If it is at $p^\dagger$ for some $\alpha'\geq \max\left\{\alphardagger, \alphaopt\right\}$, we have
\[\pi_{s,s}\left(p^\dagger, \alpha'\right)\leq \pi_{s,s}\left(p^\dagger, \max\left\{\alphardagger, \alphaopt\right\}\right)\leq \pi_{s,s}\left(\psstar, \max\left\{\alphardagger, \alphaopt\right\}\right)\]
This leaves us with the result that the independent seller's payoff in equilibrium of the sequential game is bounded by
\[\piseq(\rho(\alphastar))\leq \pi_{s,s}\left(\psstar, \max\left\{\alphardagger, \alphaopt\right\}\right)\]
which is loose for similar reasons as our bound on $\alphastar$.

\subsection{Upper bounding the retailer's payoff}
\label{appendix:upper-bounding-retailer-payoff}

Analogous to the previous subsection, we know that the upper bound for the retailer's payoff will be with the strategy profile $\underline{\rho}\equiv \max\left\{\prstar, \psind\right\}$. The retailer's maximum payoff either occurs when the independent seller fulfills demand at $\prstar$ or $\psind$, or when the retailer fulfills demand themselves at $\prstar$.

Recall that $\psind\leq \prstar \iff \alpha \leq \alpharstar$. Notice that $\pi_{r,s}(\prstar,\alpha)=\alpha\prstar q(\prstar)$ is increasing in $\alpha$ because $\prstar $ is independent of $\alpha$. Thus, it immediately follows that
\begin{equation}\max_{\alpha\in \left[\max\left\{\alphaopt, \alphardagger\right\}, \alpharstar\right]}\pireq(\alpha, c_r, c_s, \underline{\rho})= \max_{\alpha\in \left[\max\left\{\alphaopt, \alphardagger\right\}, \alpharstar\right]}\pi_{r,s}(\prstar, \alpha) = \pi_{r,s}(\prstar, \alpharstar)\label{eq:max-pireq-underlinerho-i}\end{equation}

At $\alpharstar = 1-\frac{c_s}{\prstar}$ per \eqref{eq:alpharstar-derivation} we have
\[\pi_{r,s}(\prstar, \alpharstar)=\left(1-\frac{c_s}{\prstar}\right)\prstar q(\prstar) = (\prstar-c_s)q(\prstar)>(\prstar-c_r)q(\prstar)= \pi_{r,r}(\prstar)\]
which implies the retailer always prefers to set $\alpharstar$ rather than any $\alpha\geq \alphasdagger$ where they fulfill demand themselves.

Now let's turn our attention to the other half of the interval, $\left[\alpharstar, \alphasdagger\right]$, on which the equilibrium is that the independent seller fulfills demand at $\psind$. Indeed, on this interval we claim that $\pi_{r,s}(\psind, \alpha)$ is decreasing in $\alpha$. Indeed, notice $\psind = \frac{c_s}{1-\alpha}$ is increasing in $\alpha$, and at $\alpharstar$, $\psind=\prstar>\tilde{p}$ per Lemma \ref{lemma:ptilde-ordering}. Thus, for all $\alpha\geq \alpharstar$, we have $\psind\geq \prstar>\tilde{p}$, and since $\pi_{r,s}$ is unimodal per Corollary \ref{cor:unimodality}, this is sufficient to conclude the proof of the claim. In light of this, we have
\begin{equation}\max_{\alpha\in \left[\alpharstar, \alphasdagger\right]}\pireq(\alpha, c_r, c_s, \underline{\rho})= \max_{\alpha\in \left[\alpharstar, \alphasdagger\right]}\pi_{r,s}(\psind, \alpha) = \pi_{r,s}(\psind, \alpharstar)\end{equation}
Combining this with \eqref{eq:max-pireq-underlinerho-i} is sufficient to conclude that
\[\max_\alpha \pireq(\alpha, c_r,c_s,\underline{\rho}) = \pireq(\alpharstar, c_r, c_s, \underline{\rho}) = \pi_{r,s}(\prstar, \alpharstar) = \pi_{r,s}(\psind, \alpharstar)\]

Just as the previous case, this is unfortunately not sufficient to conclude that $\alphastar\leq \alpharstar$ in general; the best we can do is conclude that $\alphastar\leq \alphasdagger$. For the independent seller, it results only in the trivial bound that $\piseq(\alpha, c_r,c_s, \rho)\geq 0$, where the bound is tight because when $\alpha = \alpharstar$ and $\rho = \underline{\rho}$, the independent seller is fulfilling demand at $\psind$ and achieves a payoff of $0$.

\end{document}